\documentclass[a4paper,USenglish,pdfa,cleveref,autoref,thm-restate]{lipics-v2021}

\nolinenumbers

\usepackage{packages}

\usepackage{macros/misc}
\usepackage{macros/rule}
\usepackage{macros/syntax}
\usepackage{macros/lfpl}

\newcommand{\FULLVERSION}{}

\author{Nathaniel Glover}{Carnegie Mellon University, Pittsburgh, PA, USA \and \url{https://www.cs.cmu.edu/~nsglover} }{nsglover@cs.cmu.edu}{https://orcid.org/0009-0006-1123-4277}{}

\author{Jan Hoffmann}{Carnegie Mellon University, Pittsburgh, PA, USA \and \url{https://www.cs.cmu.edu/~janh} }{jhoffmann@cmu.edu}{https://orcid.org/0000-0001-8326-0788}{}

\authorrunning{N. Glover and J. Hoffmann}

\Copyright{Nathaniel Glover and Jan Hoffmann}

\ifdefined\FULLVERSION
  \hideLIPIcs
\else
  \EventEditors{Claudia Faggian and Joost-Pieter Katoen}
  \EventNoEds{2}
  \EventLongTitle{41st Annual Symposium on Logic in Computer Science (LICS 2026)}
  \EventShortTitle{LICS 2026}
  \EventAcronym{LICS}
  \EventYear{2026}
  \EventDate{July 20--23, 2026}
  \EventLocation{Lisbon, Portugal}
  \EventLogo{}
  \SeriesVolume{380}
  \ArticleNo{51}
\fi

\begin{document}

\ifdefined\FULLVERSION\else
  \relatedversiondetails[cite={fullversion}]{Full Version}{https://arxiv.org/abs/2605.12893}
\fi

\supplementdetails[cite={mechanization}]{Software}{https://doi.org/10.5281/zenodo.18348212}

\title{LFPL: Revisited and Mechanized}

\maketitle

\begin{abstract}
  Hofmann (1999) introduced the functional programming language \lfpl to characterize the functions computable in polynomial time using an affine type system.
  \lfpl enables a natural programming style, including nested recursion, and has inspired the development of type systems for automatic cost analysis, linear dependent type theories, and efficient memory management in functional programming languages.
  Despite its prominence, there does not exist a self-contained presentation, let alone a full mechanization, of \lfpl and its core metatheory.
  
  This article presents a modern account and mechanization of \lfpl and its metatheory with the goal of being self-contained and accessible while streamlining the strongest-known soundness and completeness results.
  The soundness proof works with the language \lfplplus, which extends \lfpl with additional language features.
  The proof is novel, adapting a technique by Aehlig and Schwichtenberg (2002) to construct explicit polynomials that bound the cost of an \lfplplus expression with respect to a big-step cost semantics.
  The completeness proof shows that \lfpl programs can simulate polynomial-time Turing machines while only relying on restricted forms of linear functions and lists.
  It has the same structure as the original proof by Hofmann (2002) but greatly simplifies the core argument with a novel stack-like data structure that is implemented with first-class functions and lists.
  The mechanization includes the full soundness and completeness proofs, and serves as one of the first case studies of mechanized metatheory in the recently developed proof assistant Istari.
\end{abstract}

\section{Introduction}

The goal of implicit computational complexity (ICC) is to provide a programming language which characterizes a given complexity class. Concretely, this means that a procedure is representable in the language if and only if it inhabits the desired complexity class. ICC has been extensively studied for many complexity classes including P~\cite{hofmannLinearTypesNonsizeincreasing1999,hofmannTypeSystemBounded2000,hofmannLinearTypesNonsizeincreasing2003,gaboardiSoftLinearLogic2008}, PSPACE~\cite{hofmannLinearTypesNonsizeincreasing2003, gaboardiSoftLinearLogic2008,gaboardiImplicitCharacterizationPSPACE2012}, EXPTIME~\cite{hofmannStrengthNonsizeIncreasing2002}, and LOGSPACE~\cite{schoppStratifiedBoundedAffine2007, ramyaaRamifiedCorecurrenceLogspace2011}. Such languages often achieve their desired complexity class via modifications of standard linear or affine type systems~\cite{dallagoImplicitComputationComplexity2022}.

One prominent such programming language for the complexity classes P and FP is Martin Hofmann's Linear Function Programming Language (\lfpl).
It was first defined in 1999, when Hofmann~\cite{hofmannLinearTypesNonsizeincreasing1999} showed that \lfpl is polynomial-time sound, that is, all functions definable in \lfpl belong to FP.
In 2002, Hofmann~\cite{hofmannStrengthNonsizeIncreasing2002} proved the logical inverse of this statement, polynomial-time completeness.
\lfpl is notable for characterizing FP while also supporting a natural programming style and giving rise to the development of automatic amortized resource analysis~\cite{HoffmJ22}.
Recently, there has been renewed interest in \lfpl. Atkey~\cite{atkeyPolynomialTimeDependent2024} extends \lfpl with dependent
types and unrestricted computation at the type level, creating a
language equipped with the powerful guarantees of dependent type
theory while still enforcing that computation is polynomial-time at
the term level. Lorenzen et al.~\cite{lorenzenFP2FullyInPlace2023}
utilize the type system of \lfpl in their language Koka to ensure
that certain functional programs can be executed fully in-place.

Despite its prominence, we are not aware of a fully self-contained presentation of \lfpl and its core metatheory. The original presentation by Hofmann~\cite{hofmannLinearTypesNonsizeincreasing1999} contains a semantic soundness proof but only a weak completeness result for linear-space polynomial-time Turing machines. A somewhat informal completeness result for the original \lfpl is given as a corollary in Hofmann~\cite{hofmannStrengthNonsizeIncreasing2002}, where polynomial time is not the main focus. Aehlig and Schwichtenberg~\cite{aehligSyntacticalAnalysisNonsizeincreasing2002} provide a syntactic soundness proof for a small-step semantics and a weaker completeness proof for a strong extension of \lfpl. Atkey~\cite{atkeyPolynomialTimeDependent2024} presents a similar completeness proof, as well as a soundness proof for \lfpl extended with dependent type theory. Lorenzen et al.~\cite{lorenzenFP2FullyInPlace2023} present \lfpl's type system without the metatheory. None of these works contains both a soundness proof and a completeness proof of the original \lfpl without extensions.

In this article, we give a modern account and mechanization~\cite{mechanization} of \lfpl and its metatheory. Our primary goal is for this article to be a self-contained and accessible resource for those who are interested in learning about \lfpl, with all of our results guaranteed correct by the mechanization. We have witnessed firsthand the necessity of such a resource. It was challenging to collect all this information on \lfpl from across the various different papers we have mentioned, and during the process we discovered some errors in Hofmann's completeness proof~\cite{hofmannStrengthNonsizeIncreasing2002}, the only full completeness proof that we are aware of.

Our presentation of \lfpl involves a denotational semantics similar to the original given by Hofmann~\cite{hofmannLinearTypesNonsizeincreasing1999} as well as a big-step operational cost semantics, which, to our knowledge, was not developed for \lfpl. An advantage of both semantics -- which greatly reduce the complexity of our mechanization -- is that they both avoid making use of syntactic substitution, which is a notoriously tedious operation to support and reason about in proof assistants.

The mechanized completeness proof is a simplified and streamlined version of the one given by Hofmann~\cite{hofmannStrengthNonsizeIncreasing2002}. The main novelty is the implementation of a stack-like data structure in \lfpl that allows us to avoid much of the complexity of Hofmann's proof as well as sidestep the errors we discovered in that proof. To our knowledge, ours is the strongest completeness result for \lfpl. Hofmann also gave a weaker completeness theorem for linear-time Turing machines only~\cite{hofmannLinearTypesNonsizeincreasing1999}. Both Aehlig and Schwichtenberg~\cite{aehligSyntacticalAnalysisNonsizeincreasing2002} and Atkey~\cite{atkeyPolynomialTimeDependent2024} provide full completeness proofs but for powerful extensions of \lfpl, thereby weakening the completeness result. Like Hofmann's version~\cite{hofmannStrengthNonsizeIncreasing2002}, our completeness proof works with a minimal version of \lfpl, using the fewest tools to simulate a polynomial-time Turing machine.

The mechanized soundness proof is inspired by Aehlig and Schwichtenberg~\cite{aehligSyntacticalAnalysisNonsizeincreasing2002}. We give a modern presentation of their result for a big-step operational cost semantics and an extended language \lfplplus. Using this cost semantics allows us to reason about concrete program executions and avoids the substitution-related tedium that arises when mechanizing a small-step semantics. While a big-step operational cost semantics is further away from a concrete machine model, it is well-known how to link it to small-step semantics, which have been shown to be reasonable with respect to lower-level semantic models~\cite{blellochProvableTimeSpace1996,Accattoli2019}.
Notably, for each \lfplplus term, the proof constructs a concrete polynomial that bounds the execution cost of the term as defined by the cost semantics.
To our knowledge, this is the strongest known soundness result for \lfpl;
we strengthen it by extending the language with several features such
as a built-in stack data structure and lazy products.

\ifdefined\FULLVERSION
  \enlargethispage{-4.9\baselineskip}
\else
  \enlargethispage{-1\baselineskip}
\fi

In summary, we make the following contributions: the first complete presentation of the \lfpl metatheory, a novel completeness proof, a novel soundness proof, and the first complete mechanization of the \lfpl metatheory. The novelties in the proofs make them accessible and amenable to mechanization. Our desire to mechanize the \lfpl metatheory is partially responsible for driving us to discover these proofs. Our completeness proof resolves some errors we discovered in the only other completeness proof of equal strength to ours.

\section{LFPL}
\label{sec:lfpl}

In G\"odel's System T, recursion is controlled in such a way as to ensure the language is still terminating. This is known as primitive or structural recursion. 
The only way to write variable-time programs in \lfpl is using structural recursion, but that alone is not enough to achieve polynomial time. This pseudocode is an example of an exponential-time program using only structural recursion:
\begin{lstlisting}
double zero = zero
double (succ n) = succ (succ (double n))

exp zero = succ zero
exp (succ n) = double (exp n)
\end{lstlisting}

The main idea of \lfpl is to enforce non-size-increasing computation. In the above example, the exponential behavior relies on the fact that \lstinline{double} increases the size of its input. \lfpl forbids such size-increasing computations. The non-size-increasing property is enforced with an affine type system, which treats input size as a limited resource that cannot be duplicated. Formally, size is represented by a type $\tpdiam$, so that an element of type $\tpdiam$ corresponds to one unit of size. The natural number $n$ requires $n$ elements of type $\tpdiam$ to construct. These resources, which we shall henceforth call diamonds, are returned when $n$ is consumed. The program \code{double} is no longer possible to write; in the \code{succ} case, one diamond is made available to us, yet we need two diamonds due to the two calls to \code{succ}.

In contrast to other polynomial-time systems that might restrict iteration such that nested iterations are impossible, \lfpl enables a fairly natural programming style. While many polynomial-time languages are complete with respect to polynomial-time \textit{computations}, they are often unable to express many natural \textit{algorithms}. Consider this pseudocode for an inefficient identity function on lists:
\begin{lstlisting}
append nil l2 = l2
append (n :: l1) l2 = n :: append l1 l2

id l = append (append l nil) nil
\end{lstlisting}
In a system that disallows nested iterations, expressing a program like this can be difficult or impossible when the output from a call to the iterative procedure like \lstinline{append} is fed into another such call. Meanwhile, \lfpl has no such restrictions, so the above program can be directly translated into \lfpl with minimal modifications. Due to this flexibility, \lfpl supports many other natural algorithms such as linear-time list reversal, insertion sort, and integer division by a constant.

\subsection{Properties of LFPL}

\enlargethispage{1\baselineskip}

The defining property of \lfpl is that computation is non-size-increasing. Consider a list resulting from the evaluation of an \lfpl program with some lists as inputs. The non-size-increasing theorem intuitively states that the length of this output list is no greater than the sum of the lengths of the input lists. We present this in detail in \pref{sec:soundness-size}, but for now the main takeaway is that \lfpl does not support problematic size-increasing functions like \lstinline{double}. As expected of the defining property of the language, the non-size-increasing theorem is a main tool in showing that \lfpl is polynomial-time sound.

The most complete version of the polynomial-time soundness theorem states that any function on the natural numbers expressible in \lfpl is also computable by a polynomial-time Turing machine. We only prove one step in the path from \lfpl to Turing machines: In a high-level cost semantics, given an \lfpl program and some inputs with size $n$, there exists a polynomial $P : \mathbb N \to \mathbb N$ such that evaluating the program costs at most $P(n)$. Our theorem is the only step in this path that contains reasoning specific to \lfpl: Given this theorem, it is a well-known and standard procedure to justify the polynomial bound for our high-level cost semantics in terms of a low-level model such as a Turing machine~\cite{Accattoli2019,blellochProvableTimeSpace1996}.

The polynomial-time completeness theorem states the logical inverse of soundness: Any function on the natural numbers that is computable by a polynomial-time Turing machine is expressible in \lfpl. As stated, this is not quite true for \lfpl; $\code{double}$ is computable in linear time yet is not expressible in \lfpl. We can fix this by instead proving completeness for all polynomial-time computable functions that do not increase the size of their input. Perhaps surprisingly, we see in \pref{thm:completeness-general} that the more general statement of completeness is only \textit{barely} false: There is a very reasonable perspective from which \lfpl is able to express all polynomial-time computations, even the size-increasing ones.

\subsection{Syntax and Typing Rules}
\label{sec:lfpl-definition}

\begin{figure}[t!]
\hspace{-1.7em}
\begin{minipage}[t]{0.33\linewidth}
\begin{syntax}{Type}{A,B}
\sitem{\tpdiam}{diamond}
\sitem{\tpunit}{unit}
\sitem{\tpsum{A}{B}}{sum}
\sitem{\tptens{A}{B}}{tensor}
\sitem{\tparr{A}{B}}{arrow}
\sitem{\tplist{A}}{list}
\end{syntax}
\end{minipage}
\hspace{3.3em}
\begin{minipage}[t]{0.6\linewidth}
\begin{syntax}{Term}{M,N}
\sitem{x}{variable}
\sitem{\tmnull}{unit intro}
\sitem{\tminj{i}{M} \qquad\qquad (i \in \{1,2\})}{sum intro}
\sitem{\tmcase{M}{x_1}{N_1}{x_2}{N_2}}{sum elim}
\sitem{\tmpair{M_1}{M_2}}{tensor intro}
\sitem{\tmletp{M}{x_1}{x_2}{N}}{tensor elim}
\sitem{\tmlam{x}{M}}{arrow intro}
\sitem{\tmapp{M}{N}}{arrow elim}
\sitem{\tmnil}{list intro}
\sitem{\tmcons{M_d}{M_h}{M_t}}{list intro}
\sitem{\tmrec{M}{N_1}{x_d}{x_h}{x_t}{N_2}}{list elim}
\end{syntax}
\end{minipage}
\caption{\lfpl's type and term syntax.}
\label{fig:syntax-tp-tm}
\end{figure}

The syntax of \lfpl is given in \pref{fig:syntax-tp-tm}. Other than the type $\tpdiam$, and the additional argument and variable appearing in the list introduction and elimination forms respectively, the syntax looks like System T with lists. The first argument $M_d$ in the second list introduction form $\tmcons{M_d}{M_h}{M_t}$ is intended to have type $\tpdiam$; to add an element to a list, we are forced to pay a diamond. Likewise, the first variable $x_d$ in the list elimination $\tmrec{M}{N_1}{x_d}{x_h}{x_t}{N_2}$ is intended to have type $\tpdiam$; when eliminating a list, we earn back the diamonds we paid to construct it, allowing us to construct new lists during recursion. In this way, we use elements of type $\tpdiam$ to control the lengths of lists, ensuring that the number of diamonds available to us bounds the maximum length of a list that we can construct at any point.

\enlargethispage{-1\baselineskip}
The typing judgement is written $\types{\ctx}{M}{A}$ and read as ``$M$ has type $A$ under context $\ctx$''. Importantly, \lfpl has an \textit{affine} type system, so contraction is not allowed and we must partition the context when we type an expression containing sequential subexpressions. The typing rules are given in \pref{fig:jmt-ty}. Notice that for $\tmcons{M_d}{M_h}{M_t}$ to be typed under $\ctxappend{\ctx_d}{\ctxappend{\ctx_h}{\ctx_t}}$, we must have $\types{\ctx_d}{M_d}{\tpdiam}$. Likewise, notice that for $\tmrec{M}{N_1}{x_d}{x_h}{x_t}{N_2}$ to be typed, we must type $N_2$ under a context with $x_d : \tpdiam$.

Another important subtlety is in \pref{rule:ty-ListE}, where $N_2$ is typed under a context containing \textit{only} the variables $x_d$, $x_h$, and $x_t$. One is not allowed to use additional variables from the surrounding environment. This is because $N_2$ represents the body of the structural recursor and therefore, despite appearing once statically, may be run many times during the execution of a program. Consequently, even if $N_2$ would statically appear to use a variable once, at runtime that variable would be used many times and thus violate the main principle of affine type theory. For example, $N_2$ could take just one diamond from the context and use it to create a list of length $10$, which would break \lfpl's non-size-increasing property.

\NewDocumentCommand{\RuleTy}{m m m}{\Rule{\text{Ty}:#1}{#2}{#3}\label{rule:ty-#1}}
\begin{figure}[t!]
\begin{mathpar}
\RuleTy{Var}{\types{\ctxcons{\ctx}{x}{A}}{x}{A}}{}
\and
\RuleTy{UnitI}{\types{\ctx}{\tmnull}{\tpunit}}{}
\and
\RuleTy{ArrowI}{\types{\ctx}{\tmlam{x}{M}}{\tparr{A}{B}}}{\types{\ctxcons{\ctx}{x}{A}}{M}{B}}
\and
\RuleTy{ArrowE}{\types{\ctxappend{\ctx}{\ctx'}}{\tmapp{M}{N}}{B}}{\types{\ctx}{M}{\tparr{A}{B}} \\ \types{\ctx'}{N}{A}}
\and
\RuleTy{SumI$_i$}{\types{\ctx}{\tminj{i}{M}}{\tpsum{A_1}{A_2}}}{\types{\ctx}{M}{A_i}}
\and
\RuleTy{SumE}{\types{\ctxappend{\ctx}{\ctx'}}{\tmcase{M}{x_1}{N_1}{x_2}{N_2}}{B}}{\types{\ctx}{M}{\tpsum{A_1}{A_2}} \\ \types{\ctxcons{\ctx'}{x_1}{A_1}}{N_1}{B} \\ \types{\ctxcons{\ctx'}{x_2}{A_2}}{N_2}{B}}
\and
\RuleTy{TensorI}{\types{\ctxappend{\ctx_1}{\ctx_2}}{\tmpair{M_1}{M_2}}{\tptens{A_1}{A_2}}}{\types{\ctx_1}{M_1}{A_1} \\ \types{\ctx_2}{M_2}{A_2}}
\and
\RuleTy{TensorE}{\types{\ctxappend{\ctx}{\ctx'}}{\tmletp{M}{x_1}{x_2}{N}}{B}}{\types{\ctx}{M}{\tptens{A_1}{A_2}} \\ \types{\ctxcons{\ctxcons{\ctx'}{x_1}{A_1}}{x_2}{A_2}}{N}{B}}
\and
\RuleTy{ListI$_1$}{\types{\ctx}{\tmnil}{\tplist{A}}}{}
\and
\RuleTy{ListI$_2$}{\types{\ctxappend{\ctx_d}{\ctxappend{\ctx_h}{\ctx_t}}}{\tmcons{M_d}{M_h}{M_t}}{\tplist{A}}}{\types{\ctx_d}{M_d}{\tpdiam} \\ \types{\ctx_h}{M_h}{A} \\ \types{\ctx_t}{M_t}{\tplist{A}}}
\and
\RuleTy{ListE}{\types{\ctxappend{\ctx}{\ctx'}}{\tmrec{M}{N_1}{x_d}{x_h}{x_t}{N_2}}{B}}{\types{\ctx}{M}{\tplist{A}} \\ \types{\ctx'}{N_1}{B} \\ \types{\ctxcons{\ctxcons{\ctxcons{\ctxnil}{x_d}{\tpdiam}}{x_h}{A}}{x_t}{B}}{N_2}{B}}
\end{mathpar}
\caption{\lfpl's affine typing rules.}
\label{fig:jmt-ty}
\end{figure}

\subsection{Denotational Semantics}
\label{sec:denotational-semantics}

To reason about the behavior of \lfpl programs, we need to reason about evaluating functions which may take lists as input. For instance, we will see in \pref{sec:lfpl-examples} a closed \lfpl term $\code{reverse} : \tparr{\tplist{A}}{\tplist{A}}$, but we can never actually ``run'' $\code{reverse}$ on a nonempty input list, since that would require a closed term of type $\tpdiam$ and the whole point of \lfpl's type system is to make that impossible. Thus, it is difficult to study \lfpl by itself; we instead need to reason about terms in a domain where $\tpdiam$ is inhabited.

\enlargethispage{\baselineskip}
We now give a simple set-theoretic denotational semantics. To each type $A$ we assign a set $\tpsem{A}$, and to each well-typed term $\types{\Gamma}{M}{A}$ we assign an element $\tmsem{M}{\env}$ of $\tpsem{A}$, where $\env$ is an environment such that $\env(x) \in \tpsem{A}$ for every $x : A \in \Gamma$, written as $\env \in \tpsem{\Gamma}$. Per the above discussion, we define $\tpsem{\tpdiam}$ as a nonempty set to ensure $\tpdiam$ is semantically inhabited.

To define the semantics of the \lfpl list type, we make use of the Kleene star operator. Given a set $S$, the set $S^*$ is the set of finite strings of elements of $S$. We write $\emptylist$ for the empty string, and given $x \in S$ and $\ell \in S^*$, we write $x :: \ell$ for the element of $S^*$ obtained by putting $x$ before the sequence of elements in $\ell$. Given $\ell \in S^*$, let $|\ell| \in \mathbb N$ denote the length of the string $\ell$. Lastly, for every set $T$, we inductively define the iteration operator on $S^*$:
\begin{align*}
&\text{iter}(\emptylist, b, f) = b \\
&\text{iter}(x :: \ell, b, f) = f(x, \text{iter}(\ell, b, f))
\end{align*}

For the semantics of the sum type, we make use of the disjoint union (or coproduct) operator $S \sqcup T := (\{ 1 \} \times S) \cup (\{ 2 \} \times T)$. The definitions of $\tpsem{A}$ and $\tmsem{M}{\env}$ are given in \pref{fig:sem-tp-tm}, inspired by Hofmann's set-theoretic interpretation~\cite{hofmannStrengthNonsizeIncreasing2002}. Notice the argument of type $\tpdiam$ to $\code{cons}$ is ignored; the presence of $\tpdiam$ in lists is purely a means of statically enforcing polynomial-time computation, so the dynamic semantics for $\code{cons}$ need not involve $\blacklozenge$.

\begin{figure}[t!]
\vspace{-1em}
\[
\begin{array}{@{\hspace{3.3em}}r@{\;=\;}l@{\hspace{3.3em}}r@{\;=\;}l@{\hspace{3.3em}}r@{\;=\;}l}
\tpsem{\tpdiam} & \{ \blacklozenge \} &
\tpsem{\tptens{A}{B}} & \tpsem{A} \times \tpsem{B} &
\tpsem{\tplist{A}} & \tpsem{A}^* \\
\tpsem{\tpunit} & \{ \ast \} &
\tpsem{\tpsum{A}{B}} & \tpsem{A} \sqcup \tpsem{B} &
\tpsem{\tparr{A}{B}} & \tpsem{A} \to \tpsem{B}
\end{array}
\]
\[
\begin{array}{rcll}
\tmsem{x}{\env} &=& \env(x) \\
\tmsem{\tmnull}{\env} &=& \ast \\
\tmsem{\tminj{i}{M}}{\env} &=& (i, \tmsem{M}{\env}) \\
\tmsem{\tmcase{M}{x_1}{N_1}{x_2}{N_2}}{\env} &=& \tmsem{N_i}{\envcons{\env}{x_i}{v}} 
&\!\!\!\!\! \text{where }  \tmsem{M}{\env} = (i, v) \\
\tmsem{\tmpair{M_1}{M_2}}{\env} &=& (\tmsem{M_1}{\env}, \tmsem{M_2}{\env}) \\
\tmsem{\tmletp{M}{x_1}{x_2}{N}}{\env} &=& \tmsem{N}{\env[x_1 \mapsto v, x_2 \mapsto u]}
&\!\!\!\!\! \text{where }  \tmsem{M}{\env} = (v, u) \\
\tmsem{\tmlam{x}{M}}{\env} &=& \lambda \, v . \; \tmsem{M}{\envcons{\env}{x}{v}} \; \\
\tmsem{\tmapp{M}{N}}{\env} &=& \tmsem{M}{\env} \left( \tmsem{N}{\env} \right) \\
\tmsem{\tmnil}{\env} &=& \emptylist \\
\tmsem{\tmcons{M_0}{M_1}{M_2}}{\env} &=& \tmsem{M_1}{\env} :: \tmsem{M_2}{\env} \\
  \tmsem{\tmrec{M}{N_1}{x_d}{x_h}{x_t}{N_2}}{\env} &=& \text{iter}(\tmsem{M}{\env}, \tmsem{N_1}{\env}, f)\\
&& \multicolumn{2}{r}{\text{where } f(v, w) = \tmsem{N_2}{\envnil[x_d \mapsto \blacklozenge, x_h \mapsto v, x_t \mapsto w]}}
\end{array}
\]

\caption{\lfpl's denotational semantics.}
\label{fig:sem-tp-tm}
\end{figure}

The notion of size and the non-size-increasing property were originally formulated for the denotational semantics~\cite{hofmannLinearTypesNonsizeincreasing1999}. The basic idea is to define a partial function $\mathsf{size}_A : \tpsem{A} \to \mathbb N$ by induction on $A$, intuitively counting the number of diamonds within the input. For example, $\mathsf{size}_{\tpdiam}(\blacklozenge) = 1$. The reason for partiality is that $\mathsf{size}_{\tparr{A}{B}}(f)$ is defined as the maximum difference (which might not exist) between $\mathsf{size}_B(f(x))$ and $\mathsf{size}_A(x)$, where $x$ ranges over all $x \in \tpsem{A}$ such that $\mathsf{size}_A(x)$ and $\mathsf{size}_B(f(x))$ are defined. The non-size-increasing theorem then proves simultaneously that $\mathsf{size}_A(\tmsem{M}{\env})$ is defined whenever the sizes of all values in $\env$ are defined and that the size of $\tmsem{M}{\env}$ is bounded by the sum of the sizes of the values in $\env$. In this paper, we provide an equivalent definition using an operational semantics in \pref{sec:soundness-size}. It is simpler because the definition is total, avoiding the complexities of dealing with partial functions while remaining sufficient for proving soundness.

\subsection{Example Programs}
\label{sec:lfpl-examples}
\enlargethispage{2.4\baselineskip}

To build intuition about the expressivity of \lfpl, we implement some examples of functions that demonstrate how to program in the language. Firstly, given any element type $A$, we implement efficient list reversal. For readability, we use the concrete syntax
\[ \text{\lstinline{rec M | nil => N1 | cons (xd, xh, xt) => N2}} \]
for $\tmrec{M}{N_1}{x_d}{x_h}{x_t}{N_2}$ in the examples; and a similar syntax for sum elimination.

\begin{example}[Linear-Time Reverse]
\label{example:reverse}
We can implement the standard linear-time list reversal, often seen in introductory material on functional programming.
\begin{lstlisting}
revAppend : L(A) -o L(A) -o L(A)
revAppend = lam l1 . rec l1
| nil => lam l2 . l2
| cons (d, x, r) => lam l2 . r (cons (d, x, l2))

reverse : L(A) -o L(A)
reverse = lam l1 . revAppend l1 nil
\end{lstlisting}
In the abstract syntax of \pref{fig:syntax-tp-tm}, \code{revAppend} would be written as:
\[ \tmlam{\ell_1}{\tmrec{\ell_1}{\tmlam{\ell_2}{\ell_2}}{d}{x}{r}{\tmlam{\ell_2}{\tmapp{r}{(\tmcons{d}{x}{\ell_2})}}}} \]
This term has type $\tparr{\tplist{A}}{\tparr{\tplist{A}}{\tplist{A}}}$. Using the rules in \pref{fig:jmt-ty}, we can see this via the following partial derivation. The ellipses are straightforward to fill in. 
\begin{mathpar}
\ARule
{
  \types{\ctxnil}{\tmlam{\ell_1}{\tmrec{\ell_1}{\tmlam{\ell_2}{\ell_2}}{d}{x}{r}{\tmlam{\ell_2}{\tmapp{r}{(\tmcons{d}{x}{\ell_2})}}}}}{\tparr{\tplist{A}}{\tparr{\tplist{A}}{\tplist{A}}}}
}
{
  \ARule
  {
    \types{\ctxsing{\ell_1}{\tplist{A}}}{\tmrec{\ell_1}{\tmlam{\ell_2}{\ell_2}}{d}{x}{r}{\tmlam{\ell_2}{\tmapp{r}{(\tmcons{d}{x}{\ell_2})}}}}{\tparr{\tplist{A}}{\tplist{A}}}
  }
  {
    \cdots
    \\
    \cdots
    \\
    \ARule
    {
      \types{\ctxcons{\ctxcons{\ctxsing{d}{\tpdiam}}{x}{A}}{r}{\tparr{\tplist{A}}{\tplist{A}}}}{\tmlam{\ell_2}{\tmapp{r}{(\tmcons{d}{x}{\ell_2})}}}{\tparr{\tplist{A}}{\tplist{A}}}
    }
    {
      \ARule
      {
        \types{\ctxcons{\ctxcons{\ctxcons{\ctxsing{d}{\tpdiam}}{x}{A}}{r}{\tparr{\tplist{A}}{\tplist{A}}}}{\ell_2}{\tplist{A}}}{\tmapp{r}{(\tmcons{d}{x}{\ell_2})}}{\tplist{A}}
      }
      {
        \cdots
        \\
        \ARule
        {
          \types{\ctxcons{\ctxcons{\ctxsing{d}{\tpdiam}}{x}{A}}{\ell_2}{\tplist{A}}}{\tmcons{d}{x}{\ell_2}}{\tplist{A}}
        }
        {
          \ARule{\types{\ctxsing{d}{\tpdiam}}{d}{\tpdiam}}{} \\ \ARule{\types{\ctxsing{x}{A}}{x}{A}}{} \\ \ARule{\types{\ctxsing{\ell_2}{\tplist{A}}}{\ell_2}{\tplist{A}}}{}
        }
      }
    }
  }
}
\end{mathpar}

From this derivation, it is straightforward to derive $\types{\ctxnil}{\code{reverse}}{\tparr{\tplist{A}}{\tplist{A}}}$. To prove that \code{reverse} is correct, one can use the semantic clauses in \pref{fig:sem-tp-tm} to prove $\tmsem{\code{reverse}}{\envnil} : \tpsem{A}^* \to \tpsem{A}^*$ is the reversal function on $\tpsem{A}^*$.
\end{example}

\begin{example}[List Case Analysis]
\label{example:list-case}
Importantly, both for convenience and for use in \pref{example:susp}, we can inspect whether a list is empty, which is often presented as pattern matching in standard functional programming languages. For brevity in the concrete syntax, we allow for pattern matching on tuple variables at their binding sites in place of the more verbose \code{letp}.
\begin{lstlisting}
lfold : 1 + diam * A * L(A) -o L(A)
lfold = lam x . case x .
| inj1 _ => nil
| inj2 (d, x, xs) => cons (d, x, xs)

lunfold : L(A) -o 1 + diam * A * L(A)
lunfold = lam x . rec x .
| nil => inj1 <>
| cons (d, x, r) => inj2 (d, x, lfold r)
\end{lstlisting}
\end{example}

\begin{example}[List Suspension]
\label{example:susp}
\enlargethispage{.5\baselineskip}
Given a list, we can temporarily suspend its values and obtain the diamonds within the list for use elsewhere. We, of course, need those diamonds to un-suspend the list.
\begin{lstlisting}
susp : L(A) -o ((L(1) -o L(A)) * L(1))
susp = lam x . rec x .
| nil => (lam _ . nil, nil)
| cons (d, x, r) => letp (f, m) = r in
  ((lam n . case (lunfold n) .
     | inj1 _ => nil
     | inj2 (d', _, n') => cons (d', x, f n')), cons (d, <>, m))
\end{lstlisting}
This example is important for the completeness proof, so it is worth stating its semantics: If $x \in \tpsem{A}^*$, then $\tmsem{\code{susp}}{\envnil}(x) = (f, m)$, where $f(m) = x$. Note that $m$ is the unique element of $\tpsem{\tpunit}^*$ with length $|x|$, so we recover $x$ as long as we give $f$ enough diamonds.
\end{example}

\subsection{Contraction}

The defining trait of an affine type system is that contraction, the use of a variable more than once, is not permitted. The importance of disallowing contraction in \lfpl is clear; if we could use a variable of type $\tpdiam$ more than once, then we could easily create lists of arbitrary length and break the non-size-increasing property. Despite this restriction, there are types which still enjoy contraction. We define the judgement $\hspace{-0.25em}\dfree{}$ on types as follows:
\begin{mathpar}
\ARule{\dfree{\tpunit}}{}
\and
\ARule{\dfree{\tpsum{A}{B}}}{\dfree{A} \\ \dfree{B}}
\and
\ARule{\dfree{\tptens{A}{B}}}{\dfree{A} \\ \dfree{B}}
\end{mathpar}
Going by rule induction on the judgement $\dfree{A}$ for a given type $A$, it is straightforward to define a closed term $\code{dup}_A : \tparr{A}{\tptens{A}{A}}$ that returns two copies of its input, effectively implementing contraction at $A$.

A common approach to safely reintroducing contraction into affine type systems is the exponential modality $!$. The syntax and typing rules for it might look something like:
\begin{mathpar}
\ARule{\types{!\ctx}{\code{wrap}(M)}{\; !A}}{\types{!\ctx}{M}{A}}\hspace{-1.3em}
\and
\ARule{\types{\ctx}{\code{unwrap}(M)}{A}}{\types{\ctx}{M}{\; !A}}\hspace{-1.3em}
\and
\ARule{\types{\ctxcons{\ctx}{x}{\; !A}}{[x/y]M}{B}}{\types{\ctxcons{\ctxcons{\ctx}{x}{\; !A}}{y}{\; !A}}{M}{B}}\hspace{-1.3em}
\end{mathpar}

By $!\ctx$ we mean the context obtained by replacing $x : A$ with $x : \; !A$ for every $x$ in $\ctx$. Commonly, this modality would also contain a rule for weakening, but LFPL is affine and thus already supports weakening at every type. It is tempting to add this modality to \lfpl; if we say that all values of type $!A$ have size $0$, then these operations are non-size-increasing. Unfortunately, that is not enough to guarantee polynomial-time soundness. Consider the following pseudocode for \lfpl (extended with the !-modality, where $\tpsem{!A} = \tpsem{A}$):
\begin{lstlisting}
fnExp : L(1) -o !(L(1) -o L(1)) -o !(L(1) -o L(1))
fnExp = lam l1 . lam f . rec l1
| nil => f
| cons (_, _, r) => wrap (lam x . unwrap(r) (unwrap(r) x))
\end{lstlisting}
Then, $\tmsem{\code{fnExp}}{\envnil}(n)(f)$ is the function $f^{2^{|n|}}$ that calls $f$ on its input $2^{|n|}$ times. Thus, we can use \code{fnExp} to implement functions that are not polynomial-time computable.

\section{Polynomial-Time Completeness}
\label{sec:completeness}

We now use the denotational semantics to state what it means for \lfpl to be complete with respect to polynomial-time computation on the natural numbers. The interpretation $\tpsem{\tplist{\tpunit}} = \{\ast\}^*$ of unit lists can be viewed as the set of unary-encoded natural numbers. However, it is also possible to use the binary representation $\tpsem{\tplist{\tpsum{\tpunit}{\tpunit}}}$ for $\mathbb N$.
Instead of choosing a particular encoding of the natural numbers, we
assume the digits of a natural number are encoded by a type $A$ such that
$\dfree{A}$. We require $\dfree{A}$ because plenty of polynomial-time
computable functions (even non-size-increasing ones) duplicate the
digits of their input, which is impossible in \lfpl if $A$ contains
diamonds.

\begin{theorem}[Size-Restricted Polynomial-Time Completeness]
\label{thm:completeness-concrete}
Suppose $\dfree{A}$. Then, for every polynomial-time computable function $f : \tpsem{A}^* \to \tpsem{A}^*$ such that $|f(x)| \le |x|$ for all $x \in \tpsem{A}^*$, there exists a closed term $M : \tparr{\tplist{A}}{\tplist{A}}$ such that $\tmsem{M}{\envnil} = f$.
\end{theorem}

This theorem was originally proven for $A = \tpunit$~\cite{hofmannStrengthNonsizeIncreasing2002}. The constraint that $f$ does not increase the length of its input is necessary because of the non-size-increasing property enforced upon LFPL functions. As we see in \pref{sec:completeness-general-thm}, this constraint can be removed if we do not enforce that the output type of $M$ is an \lfpl list.

Proving \pref{thm:completeness-concrete} boils down to simulating the polynomial-time Turing machine that computes the function $f$, say with polynomial bound $P : \mathbb N \to \mathbb N$. This immediately presents two challenges. The first challenge is to represent the tape of the machine, which for an input list of size $n$ could contain up to $P(n) + n$ values by the time the machine halts. We only have $n$ diamonds from the input list, yet we have to store up to $P(n) + n$ elements.

The second challenge is to simulate the Turing machine for $P(n)$ steps. At first glance, it seems that \lfpl is only suited for performing $kn$ iterations given $n$ diamonds and some constant $k$. This can be achieved for $k = 1$ by iterating down a list and building it back up during the iteration process, and then it is easy to repeat this for any constant $k$ number of times. This is useful but does not help to perform, say, $n^2$ iterations.

Both of these challenges are somewhat unique to \lfpl's completeness
proof compared to completeness proofs for other languages for P.
For example, the challenge in the proof of the Cons-free system from
the work of Jones~\cite{Jones2001} lies mostly in simulating the
polynomially many iterations of the machine’s step function but not in
providing sufficient space for the tape.
This is similar to the completeness proof for \lfplplus and the
extension of \lfpl considered by
Atkey~\cite{atkeyPolynomialTimeDependent2024}, which provide
lists without size restrictions or support for iteration.

\subsection{The Bounded Stack Data Structure}
\label{sec:completeness-bounded-stack-definition}

We first address the challenge of simulating the tape of the Turing machine. A common approach to representing the tape of the Turing machine in a functional language is to use a data structure consisting of a head element and two stacks, each representing one side of the tape relative to the head element~\cite{Okasaki99}. If we use the type $A$ for tape symbols, then the type of our tape data structure could be $\tptens{\tplist{A}}{\tptens{A}{\tplist{A}}}$. However, in \lfpl, we cannot actually use the above data structure.
Given $n$ diamonds, available in an input list of length $n$, it can only hold $n + 1$ values. Instead, we construct a type $S(A)$ that depends on the bounding polynomial $P$ and supports a stack-like interface, as long as it does not have to hold more than $P(n)$ elements. Then, we use a value of type $\tptens{S(A)}{\tptens{A}{S(A)}}$ to encode the tape.

We have in fact already seen the first key insight into how to construct $S(A)$ in \pref{example:susp}. It shows that we do not need any diamonds to store data; we just need them to interact with it. Using this idea, we can design our data structure so that it does not store diamonds and its push/pop functions temporarily borrow diamonds to interface with the suspended data.

The second key insight is that we do not need to un-suspend the entire data structure to interface with it.
Since both push and pop work at the front of the stack, we just need to be able to un-suspend enough of our data structure to expose the front.
The number of diamonds we need for these operations depends on $P$ and $n$.

\begin{definition}[Bounded Stack Interface]
Given $k \in \mathbb N$ and an element type $A$, a $k$-stack implementation is a type $S$ together with the following closed terms:
\begin{itemize}
  \item $\code{empty} : S$
  \item $\code{push} : \tparr{(\tplist{\tpunit})^k}{\tparr{\tptens{A}{S}}{\tptens{(\tplist{\tpunit})^k}{(\tptens{S}{(\tpsum{A}{\tpunit})})}}}$
  \item $\code{pop} : \tparr{(\tplist{\tpunit})^k}{\tparr{S}{\tptens{(\tplist{\tpunit})^k}{(\tptens{S}{(\tpsum{\tpunit}{A})})}}}$
\end{itemize}
\end{definition}

By $(\tplist{\tpunit})^k$ we mean the type of $k$-tuples $\tplist{\tpunit} \otimes \cdots \otimes \tplist{\tpunit}$. The diamonds stored in each list are used to un-suspend parts of the data structure and returned after the respective operation. The reason for using a $k$-tuple of lists instead of just one big list becomes apparent in the proof of \pref{lem:stack-inductive}, where the use of such a tuple greatly simplifies the implementation.
Both $\code{push}$ and $\code{pop}$ also return a stack and a sum. The left branch of the sum is intended to signal failure of the operation (in which case the input stack is returned unmodified), and the right branch signals success (in which case the returned stack has been successfully modified).
For $\code{push}$, which takes in an element and a stack, failure means that the input stack is full, and so the input element is returned to us. For $\code{pop}$, failure means the input stack is empty.

This interface is inspired by the array-like data structure used in Hofmann's original completeness proof~\cite{hofmannStrengthNonsizeIncreasing2002}. The key difference is that his array data structure has read and write operations requiring an index, given as a binary \lfpl number (i.e., $\tplist{\tpsum{\tpunit}{\tpunit}}$). The extra work of dealing with this index significantly complicates both the implementation and use of the data structure.
Arrays are more general, but our stack is easier to implement and sufficient for representing a Turing machine's tape, so we believe it is the better choice of data structure for the completeness theorem.

To formalize our intuition of how $\code{push}$ and $\code{pop}$ should behave, we use the denotational semantics. We first define the element $m_{n,k} \in \tpsem{(\tplist{\tpunit})^k}$ to be the tuple $(\ell, \ldots, \ell)$, where $\ell$ is the unique element of $\tpsem{\tplist{\tpunit}}$ such that $|\ell| = n$. This way, $m_{n,k}$ contains exactly $nk$ diamonds.

\begin{definition}[Bounded Stack Correctness]
\label{def:stack-correctness}
Let $k \in \mathbb N$. Suppose $(S, \code{empty}, \code{push}, \code{pop})$ is a $k$-stack implementation with element type $A$. Given a function $B : \mathbb N \to \mathbb N$, we say this implementation is correct and bounded by $B$ if the following holds for all $n \in \mathbb N$:
\begin{itemize}
  \item There exists a relation $V_n \subseteq \tpsem{S}$ of valid stack states.
  \item There exists a function $I_n : \tpsem{S} \to \tpsem{A}^*$ that returns the list of items in the stack.
  \item Let $s = \tmsem{\code{empty}}{\ctxnil}$. Then, $V_n(s)$ holds and $I_n(s) = \emptylist$.
  \item Let $f = \tmsem{\code{push}}{\ctxnil}$. Then, for all $x \in \tpsem{A}$ and $s \in \tpsem{S}$ such that $V_n(s)$ holds, we have:
  \begin{itemize}
    \item If $|I_n(s)| = B(n)$ then $f(m_{n,k})(x, s) = (m_{n,k}, s, (1, x))$.
    \item If $|I_n(s)| < B(n)$ then $f(m_{n,k})(x, s) = (m_{n,k}, s', (2, \ast))$, where $V_n(s')$ holds and $I_n(s') = x :: I_n(s)$.
  \end{itemize}
  \item Let $f = \tmsem{\code{pop}}{\ctxnil}$. Then, for all $s \in \tpsem{S}$ such that $V_n(s)$ holds:
  \begin{itemize}
    \item If $I_n(s) = \emptylist$ then $f(m_{n,k})(s) = (m_{n,k}, s, (1, \ast))$.
    \item If $I_n(s) = x :: \ell$ then $f(m_{n,k})(s) = (m_{n,k}, s', (2, x))$, where $V_n(s')$ holds and $I_n(s') = \ell$.
  \end{itemize}
\end{itemize}
\end{definition}

The intended use of the $k$-stack data structure is to choose $n \in \mathbb N$ and always pass $m_{n,k}$ to the $\code{push}$ and $\code{pop}$ operations.
\pref{def:stack-correctness} ensures that, in this use case, the stack can hold at most $B(n)$ items. Behavior is otherwise undefined. Whenever $\code{push}$ and $\code{pop}$ receive $m_{n,k}$ as input, they return $m_{n,k}$ as the first element of their output, thus making it available for future stack operations.

\subsection{Bounded Stack Implementations}
\label{sec:completeness-bounded-stack-implementation}

We now provide an implementation of a bounded stack and ways to build more implementations on top of pre-existing ones.
This ultimately leads to a stack bounded by our choice
of a polynomial $P$, which is exactly what we need to represent the
tape of a polynomial-time Turing machine.
While most of the proofs in this section focus on the high-level
intuition and avoid writing or reasoning about concrete \lfpl
programs, these programs and other details are formalized in the mechanization~\cite{mechanization}.

\begin{lemma}[Constant Stack Construction]
\label{lem:stack-constant}
For any element type $A$ and $c \in \mathbb N$, there exists a correct $0$-stack implementation which is bounded by $B(n) = c$.
\end{lemma}
\begin{proof}
We take the implementation type to be $S = (1 + A)^c$. Intuitively, a left injection represents an open space in the stack, and a right injection represents a stored element. A stack $(x_1, \ldots, x_c) \in \tpsem{S}$ is considered valid when there is some index $0 \le j \le c$ such that $x_i = (1, \ast)$ for $1 \le i \le j$, and $x_i$ is a right injection for $j < i \le c$. 

To find the head of the stack, find the first $i$ such that $x_i = (2, a)$ for some $a \in \tpsem{A}$; the value $a$ is the head element. If $i$ is $1$, then the stack is full. If no such $i$ exists, the stack is empty. Since $c$ is a constant, this procedure of locating the stack head does not require iteration, so no diamonds are required for $\code{push}$ and $\code{pop}$.
\end{proof}

The proof of \pref{lem:stack-constant} already contains intuition for how to construct stacks with non-constant bounds. 
In the proof of \pref{lem:stack-inductive}, instead of storing data in a fixed-size tuple, we store it in a list. That list can be ``suspended'' as seen in \pref{example:susp}, so that it does not require any diamonds to store.

\begin{lemma}[Inductive Stack Construction]
\label{lem:stack-inductive}
Fix an element type $A$ and some $k \in \mathbb N$. Suppose $(S, \code{empty}, \code{push}, \code{pop})$ is a $k$-stack implementation with element type $A$ that is correct and bounded by $B : \mathbb N \to \mathbb N$. Then, there is a $(k + 1)$-stack implementation with element type $A$ that is correct and bounded by $B'(n) = n B(n)$. 
\end{lemma}
\begin{proof}
We choose $S' = \tparr{\tplist{\tpunit}}{\tplist{S}}$, the type of suspended lists with element type $S$ (called ``sub-stacks''), to be our implementation type. Let $\ell$ be the unique element of $\tpsem{\tplist{\tpunit}}$ with $|\ell| = n$. To unsuspend a stack $s \in \tpsem{S'}$ we just call $s(\ell)$. We say $s$ is valid whenever $s(\ell) = s_1 :: \cdots :: s_n :: \emptylist$, where $s_i$ is valid for all $1 \le i \le n$ and there exists some index $0 \le j \le n$ such that $s_i$ is empty for $1 \le i < j$, the substack $s_j$ is unconstrained, and $s_i$ is full for $j < i \le n$. The items $I'(S')$ are the appended item lists $I(s_i)$ of each sub-stack.

It remains to correctly implement the three stack operations. Let us call them $\code{empty}'$, $\code{push}'$, and $\code{pop}'$. The operation $\code{empty}'$ is the function of type $\tparr{\tplist{\tpunit}}{\tplist{S}}$ which always returns the empty list.
When implementing $\code{push}'$, we are given as our source of diamonds an input of type $(\tplist{\tpunit})^{k+1} = \tptens{\tplist{\tpunit}}{(\tplist{\tpunit})^k}$, so we can use the first element of this tuple, call it $\ell$, to un-suspend the input stack $s \in \tpsem{S'}$ and obtain a list of sub-stacks $s(\ell) = s_1 :: \cdots :: s_n :: \emptylist$. Starting from the last element of this list, $s_n$, attempt pushing to it using the function $\code{push}$. If this succeeds, we are done. Otherwise, $s_n$ must be full, so we should try $s_{n-1}$. We can do this all the way up to $s_1$; if push still fails on $s_1$ then the whole stack must be full. This process can be implemented using \lfpl's list recursor. Afterwards, we use \pref{example:susp} to re-suspend the modified list and return the results.

The implementation of $\code{pop}'$ is similar, but instead we start from $s_1$. If popping $s_1$ fails, it must be empty, so we should try $s_2$. We go all the way down the list, and if popping $s_n$ still fails, then the whole stack must be empty. Note this is nontrivial to implement with the list recursor, since starting from $s_1$ and going down to $s_n$ is the reverse of the recursor's execution order. We solve this with a similar version of the trick used in \pref{example:reverse}.
\end{proof}

Now that we have established the main result of this subsection, a straightforward induction combines the above two lemmas, producing stack implementations with monomial bounds.

\begin{lemma}[Monomial Stack Construction]
\label{lem:stack-monomial}
Given $c,k \in \mathbb N$, there exists a $k$-stack implementation with element type $A$ that is correct and bounded by $B(n) = cn^k$.
\end{lemma}

Since every polynomial $P : \mathbb N \to \mathbb N$ is eventually bounded by some monomial $B(n) = cn^k$, it is technically enough to stop here, which is what Hofmann's original proof does~\cite{hofmannStrengthNonsizeIncreasing2002}. However, this introduces an annoying edge case when $n = 0$ because $P(0)$ could be nonzero but $B(0) = 0$ for any choice of $c$ and $k$. Due to the complexity of the array data structure, we suspect it is not worth the trouble to avoid this problem. Our simpler stack data structure allows us to sidestep this annoyance and combine monomially bounded stacks to get polynomially bounded stacks with a few more relatively straightforward implementations.

\begin{lemma}[Weakened Stack Construction]
\label{lem:stack-weakened}
Fix an element type $A$ and some $k \in \mathbb N$. Suppose $(S, \code{empty}, \code{push}, \code{pop})$ is a $k$-stack implementation with element type $A$ that is correct and bounded by $B : \mathbb N \to \mathbb N$. Then, there is a $(k + 1)$-stack implementation with element type $A$ that is correct and bounded by $B$. 
\end{lemma}
\begin{proof}
This construction is relatively obvious; if $\code{push}$ and $\code{pop}$ need only $kn$ diamonds, accepting $(k + 1)n$ diamonds instead is no problem. We can reuse the implementation type $S$ as well as the code for $\code{empty}$. For the new stack operations $\code{push}'$ and $\code{pop}'$, they simply make use of the fact that $m_{n,k+1} = (\ell,m_{n,k})$ for some $\ell \in \tpsem{\tplist{\tpunit}}$. They just ignore $\ell$ and call $\code{push}$ and $\code{pop}$ (respectively) with $m_{n,k}$.
\end{proof}

The following lemma, which shows how to combine two stacks into one stack bounded by the sum of their bounds, is the main tool by which we are able to construct stacks with arbitrary polynomial bounds from stacks with monomial bounds. It is straightforward to do this construction for our stack data structure, but it is not clear how to do it for Hofmann's array data structure, 
which only yields monomially bounded arrays~\cite{hofmannStrengthNonsizeIncreasing2002}.

\begin{lemma}[Additive Stack Construction]
\label{lem:stack-additive}
Fix an element type $A$ and some $k \in \mathbb N$. For $i \in \{1,2\}$, suppose $(S_i, \code{empty}_i, \code{push}_i, \code{pop}_i)$ is a $k$-stack implementation with element type $A$ that is correct and bounded by $B_i : \mathbb N \to \mathbb N$. Then, there is a $k$-stack implementation with element type $A$ that is correct and bounded by $B(n) = B_1(n) + B_2(n)$. 
\end{lemma}
\begin{proof}
We take the implementation type to be $S = \tptens{S_1}{S_2}$. A stack $(s_1, s_2) \in \tpsem{S}$ is considered valid when $s_1$ and $s_2$ are valid, and $s_1$ is only nonempty when $s_2$ is full. Intuitively, $s_1$ contains the first section of our stack and $s_2$ contains the second. The operations are very similar to that of \pref{lem:stack-constant} with $c = 2$. To pop the stack, we first try popping $s_1$. If that fails, it must have been empty, so we pop $s_2$. To push some $x \in \tpsem{A}$ onto this stack, we first try pushing to $s_2$. If that fails, it means $s_2$ is full, so we push to $s_1$.
\end{proof}

\begin{lemma}[Polynomial Stack Construction]
\label{lem:stack-polynomial}
Given a polynomial $P : \mathbb N \to \mathbb N$ with degree $k$, there exists a $k$-stack implementation with element type $A$ that is correct and bounded by $P$.
\end{lemma}
\begin{proof}
Write $P(n)$ as $\sum_{i=0}^k c_i n^i$. For each $0 \le i \le k$, we can use \pref{lem:stack-monomial} to obtain an $i$-stack with bound $B_i(n) = c_i n^i$. Then, we repeatedly use \pref{lem:stack-weakened} to weaken all of them into $k$-stacks with the same bounds. Finally, we repeatedly use \pref{lem:stack-additive} to obtain the desired $k$-stack whose bound is the sum of the bounds of the weakened monomial stacks.
\end{proof}

\subsection{Polynomial Iteration}
\label{sec:completeness-poly-iter}

We now address the challenge of simulating the iteration of the transition function of a $P$-time Turing machine for $P(n)$ steps when given an input list of length $n$. Additionally, we need to simulate the transition function itself. Fortunately, both of these tasks require comparably simpler solutions than the issue of polynomial storage presented in the previous section. To iterate a function on itself $P(n)$ times, we use the next two lemmas, which are inspired by Proposition 4.1 in Hofmann's work~\cite{hofmannLinearTypesNonsizeincreasing1999}.

\begin{lemma}[Linear Iteration]
\label{lem:linear-iterator}
Let $f : \tparr{\tptens{A}{\tplist{\tpunit}}}{\tptens{A}{\tplist{\tpunit}}}$ be a closed term. Then, there exists a closed term $f^{\sharp} : \tparr{\tptens{A}{\tplist{\tpunit}}}{\tptens{A}{\tplist{\tpunit}}}$ such that $\tmsem{f^{\sharp}}{\envnil} (x, n) = (\tmsem{f}{\envnil})^{|n|}(x, n)$.
\end{lemma}
\begin{proof}
Consider the following pseudocode:
\begin{lstlisting}
g : L(1) -o A * L(1) -o A * L(1)
g = lam m . rec m
| nil => lam s . s
| cons (d, u, r) => lam (x, n) . f (r (x, cons (d, u, n)))
\end{lstlisting}
We can then define $f^{\sharp} = \tmlam{s}{\tmletp{s}{x}{n}{\tmapp{\tmapp{g}{n}}{\tmpair{x}{\tmnil}}}}$.
\end{proof}

\begin{lemma}[Polynomial Iteration]
\label{lem:polynomial-iterator}
Let $f : \tparr{\tptens{A}{\tplist{\tpunit}}}{\tptens{A}{\tplist{\tpunit}}}$ be a closed term and $P : \mathbb N \to \mathbb N$ be a polynomial. Then, there exists a closed term $f^{\sharp P} : \tparr{\tptens{A}{\tplist{\tpunit}}}{\tptens{A}{\tplist{\tpunit}}}$ such that $\tmsem{f^{\sharp P}}{\envnil} (x, n) = (\tmsem{f}{\envnil})^{P(|n|)}(x, n)$.
\end{lemma}
\begin{proof}
For monomials, we can repeat the $\sharp$ operation from \pref{lem:linear-iterator}, e.g., $\tmsem{(f^{\sharp})^{\sharp}}{\envnil} (x, n) = (\tmsem{f}{\envnil})^{|n|^2}(x, n)$. To add these monomials together to obtain any polynomial, we can just use function composition, since $g^m \circ g^k = g^{m + k}$ for any $g : X \to X$ and $m,k \in \mathbb N$.
\end{proof}


We now address the task of encoding the transition function of the Turing machine.
Suppose we are dealing with a tape alphabet that includes one blank symbol $a_0$ and the rest of the symbols are drawn from the set $\tpsem{A}$, where $\dfree{A}$. Additionally, suppose the set of internal states of the Turing machine includes one halting state $q_0$ and the rest of the states are drawn from the set $\tpsem{Q}$, where again $\dfree{Q}$. Then, the transition function $g$ of the Turing machine has the following signature:
\[ g : \tpsem{Q} \times (\{ a_0 \} \sqcup \tpsem{A}) \to (\{ q_0 \} \sqcup \tpsem{Q}) \times (\{ a_0 \} \sqcup \tpsem{A}) \times \{ \Leftarrow, \Rightarrow \} \]

If $(q', a', d) = g(q, a)$, the intent is that $q$ is the current state of the Turing machine, $a$ is the symbol under the tape head, $q'$ is the new state after taking a step, $a'$ is the symbol to overwrite $a$ with, and $d$ is the direction (left or right) to shift the head. So, the type of the \lfpl term encoding $g$ should be:
\[ \tparr{\tptens{Q}{(\tpsum{\tpunit{}}{A})}}{\tptens{(\tpsum{\tpunit{}}{Q})}{\tptens{(\tpsum{\tpunit{}}{A})}{(\tpsum{\tpunit}{\tpunit})}}} \]
The input and output sets of $g$ are both denotations of $\tpdiam$-free types. The following two lemmas show \lfpl is complete with respect to computations only involving types $A$ such that $\dfree{A}$, and thus that $g$ can be implemented.

\begin{lemma}
Suppose $\dfree{A}$. Then, for every $a \in \tpsem{A}$, there exists a closed term $M_a : A$ such that $\tmsem{M_a}{\envnil} = a$.
\end{lemma}

\begin{lemma}
Suppose $A$ and $B$ are types such that $\dfree{A}$ and, for every $b \in \tpsem{B}$, that there exists a closed term $M_b : B$ such that $\tmsem{M_b}{\envnil} = b$. Then, for every $f : \tpsem{A} \to \tpsem{B}$, there exists a closed term $M_f : \tparr{A}{B}$ such that $\tmsem{M_f}{\envnil} = f$.
\end{lemma}
Both of these are straightforward to prove by rule induction on $\dfree{A}$. By combining the two, we obtain a closed term $M_g$ encoding the transition function $g$ of the Turing machine.

\subsection{Simulating the Turing Machine}
\label{sec:completeness-general-thm}
\enlargethispage{\baselineskip}

It only remains to put everything together. For this subsection, assume we have a function $f : \tpsem{A}^* \to \tpsem{A}^*$ computable by a Turing machine on a tape with alphabet $\{a_0\} \sqcup \tpsem{A}$, with internal state set $\{q_0\} \sqcup \tpsem{Q}$ and transition function $g$. Further assume the machine always halts in $P(n)$ steps on an input of length $n$, where $P : \mathbb N \to \mathbb N$ is a $k$-degree polynomial. Lastly, assume $\dfree{A}$ and $\dfree{Q}$. This assumption is not restrictive; Turing machines have finite alphabets and finitely many states.

Our goal is to construct a closed term $M : \tplist{A} \to \tplist{A}$ such that $\tmsem{M}{\envnil} = f$. Assuming an input of $\ell \in \tpsem{A}^*$ with length $n = |\ell|$, the body of the function for $M$ only has access to $n$ diamonds. We are going to need a tape data structure that can hold the input list and any additional memory used by the Turing machine, so up to $n + P(n)$ values, while only requiring $n$ diamonds. 

Since the Turing machine receives its input as a single list, giving us a pool of diamonds in the form $\tplist{\tpunit}$, yet stack operations require diamonds in the form $(\tplist{\tpunit})^{k + 1}$, we need a way to losslessly convert between diamonds collected into a list of type $\tplist{\tpunit}$ and diamonds in a divided form $\tptens{(\tplist{\tpunit})^{k + 1}}{\tplist{\tpunit}}$. The reverse direction is straightforward; just append all the lists together into one. The forward direction amounts to implementing unary division by the constant $k + 1$ in \lfpl.
This is nontrivial but not too different from an implementation of $\code{divmod}_k : \mathbb N \to \mathbb N \times \mathbb N$ in a standard functional programming language, where $\code{divmod}_k(n)$ returns the quotient and remainder of $n$ by $k + 1$.

Consider the polynomial $P'(m) = (k + 1)(m + 1) + P((k + 1)(m + 1))$, which has degree at most $\max(1, k)$, which we further bound with $k + 1$ to avoid a degree-zero edge case. Let $(m, r) = \code{divmod}_k(n)$ so that $0 \le r \le k$ and $n = m(k + 1) + r$. Then, $(k + 1)(m + 1) = m(k + 1) + (k + 1) > m(k + 1) + r = n$, implying $P'(m) \ge P(n) + n$. Therefore, we can apply \pref{lem:stack-polynomial} to $P'$ to obtain a stack implementation for element type $\tpsum{\tpunit{}}{A}$, say with underlying type $S$. These stacks hold up to $P'(m)$ values assuming we have $m(k + 1)$ diamonds available  and therefore at least $P(n)$ values. So we can model the tape with a data structure of type $\tptens{S}{\tptens{(\tpsum{\tpunit{}}{A})}{S}}$, as discussed earlier.

This is another place where our stack data structure simplifies the argument from Hofmann's array-based tape implementation~\cite{hofmannStrengthNonsizeIncreasing2002}. The array-based tape requires an index representing what element of the array is currently beneath the head of the tape. This index needs additional diamonds to store, creating the additional complexity of properly sharing the total $n$ diamonds between the index and the tape operations.

We obtain the following more general completeness theorem alluded to earlier, which
drops the size restriction of \pref{thm:completeness-concrete}.
\begin{theorem}[Polynomial-Time Completeness]
\label{thm:completeness-general}
Let $A$ be a type such that $\dfree{A}$ and let $P : \mathbb N \to \mathbb N$ be a polynomial with degree $k$. For every $P$-time computable function $f : \tpsem{A}^* \to \tpsem{\tpsum{\tpunit{}}{A}}^*$, there exists a correct $(k+1)$-stack implementation with element type $\tpsum{\tpunit{}}{A}$, say with underlying type $S$ and item functions $I_n : \tpsem{S} \to \tpsem{\tpsum{\tpunit{}}{A}}^*$, and a closed term $M : \tparr{\tplist{A}}{S}$ such that for every $x \in \tpsem{A}^*$:
\[ f(x) = I_{|x|}\left(\tmsem{M}{\envnil}(x)\right) \]
In other words, $M$ outputs a stack with items equal to $f(x)$.
\end{theorem}
\begin{proof}
To define $M$, we start with a function abstraction taking in the input $x : \tplist{A}$. Using the tape data structure discussed above, and the $n$ diamonds from $x$, we can write the values of $x$ onto the tape. Then, using the tools from \pref{sec:completeness-poly-iter}, we take an \lfpl implementation of the transition function $g$ and iterate it on the tape (as well as an \lfpl encoding of the initial machine state $q_0 \in \tpsem{Q}$) for a total of $P(n)$ steps. Finally, we return the stack representing the right half (or left; this choice just depends on the particulars of how we define a Turing machine's output) of the tape.
\end{proof}

By allowing $M$ to output a stack rather than insisting on a list, we can implement \textit{any} polynomial-time computation, even the size-increasing ones. Since the stack operations from \pref{sec:completeness-bounded-stack-definition} and the polynomial iterator from \pref{lem:polynomial-iterator} do not permanently consume any diamonds, we still have $n$ diamonds left over after the procedure in the above proof. We can just use those to remove the first $n$ values of the stack, store them in a list, and filter out the values representing blank symbols, thereby proving \pref{thm:completeness-concrete} as a corollary of \pref{thm:completeness-general}.

\subsection{Errors in the Original Proof}
\label{sec:completeness-errors}

We originally intended to mechanize Hofmann’s completeness proof~\cite{hofmannStrengthNonsizeIncreasing2002}, but thinking carefully about some of the details, we noticed several subtle errors and wanted to avoid the difficulties that come with resolving them. Thus, the mechanization effort led us to invent a simpler stack-like data structure for representing the tape of a Turing machine.

A main difficulty in the completeness proof is to represent the tape of a polynomial-time Turing machine in LFPL, say with polynomial $P$ and degree $k$. Hofmann constructs a data structure (called a “storage device”) that can hold only $cn^k$ elements for a fixed $c$, as opposed to $P(n)$ elements, given temporary access to $n$ diamonds. In Hofmann’s encoding of the Turing machine, given an input list of length n, the first step is to construct a storage device capable of holding $P(2km)$ elements when given access to $km$ temporary diamonds, which Hofmann does by finding a monomial $cm^k \ge P(2km)$. Then, $m$ is taken to be $n/(2k)$, but the largest possible value for $m$ is $\left\lfloor n/(2k) \right\rfloor$; any larger and we could potentially have $km > n/2$. This makes the chosen storage device unsuitable for the proof since it would mandate the availability of more than our $n/2$ diamonds. See below for why we have only $n/2$ rather than $n$ diamonds. So, due to this upper bound on $m$, to have $cm^k \ge P(n)$ we actually need $cm^k$ to bound $P(2k(m + 1))$, which is one small error in the proof.

More importantly, there’s the issue that $m = \left\lfloor n/(2k) \right\rfloor = 0$ when $n < 2k$, and thus $cm^k = 0$ while it may be that $P(0) > 0$, so Hofmann’s construction fails for input lists of length less than $2k$ (as well as any inputs of length $0$). This can be fixed by hard-coding each of these finitely many inputs. Nonetheless, neither this problem nor any potential solutions are discussed in the original proof. Our proof sidesteps this issue since our stack construction works for general polynomial bounds.

The second source of errors is that storage devices have an array-like interface, where the client can access any element by providing its index as a binary natural number. In \lfpl, such an index requires some diamonds to store, whereas our stack-like interface does not have indices to worry about. In Hofmann’s encoding of the Turing machine, the input list (say with length $n$) is immediately split into two lists of length $n/2$. One half is reserved for the storage device as discussed in the previous paragraphs. The purpose of the other half is not discussed in the proof, but it is necessary to store a storage device index representing the position of the tape head. The issue is that the tape could hold up to $P(n)$ elements, but the remaining diamonds can only represent up to $2^{n/2}$ binary indices, which is not necessarily greater than $P(n)$. This could be fixed by hard-coding more inputs or storing the index as a $d$-ary number for sufficiently large $d$. None of this is discussed in Hofmann’s proof, and fortunately we are once again able to sidestep all of it with our stack data structure.

\section{Polynomial-Time Soundness}
\label{sec:soundness}
\enlargethispage{\baselineskip}

This section shows that \lfpl programs can be evaluated in polynomial
time.
To strengthen the soundness result and illustrate how to add features to \lfpl while maintaining polynomial-time soundness, we extend the language and refer to the result as \lfplplus. The features are the standard product type, which manifests as the lazy product in an affine setting, a type of unbounded stacks, which are similar to lists but do not support recursion, and the inductive type of binary trees. The syntax and typing rules of these language constructs are given in \pref{fig:ext-syntax-tp-tm} and \pref{fig:ext-jmt-ty} respectively.

The typing rules for stacks of type $\tpstack{A}$ correspond to those for lists in a standard non-affine functional programming language, where the elimination form is case analysis rather than structural recursion and the introduction form does not require diamonds.
These stacks do not violate \lfpl's core principle of controlling recursion by diamonds because they do not have a recursor.

The typing rules for the tree type $\tptree{A}$, an inductive type which does support recursion, are similar to those of lists $\tplist{A}$. The primary difference is that, in \pref{rule:ty-TreeE}, the base case $N_1$ is typed under an empty context. Like the list recursor, the tree recursor has two cases: the inductive case and base case. The difference is that for the list recursor, the base case $N_1$ is only run once, so it may have access to part of the context (though this is not necessary for completeness). In the tree recursor, $N_1$ may be run many times, as a tree can have many leaves. So, to respect linearity, it cannot use anything from the context.

The denotational semantics for these new constructs are relatively straightforward. We set $\tpsem{\tpprod{A}{B}} = \tpsem{A} \times \tpsem{B}$ and $\tpsem{\tpstack{A}} = \tpsem{A}^*$. Giving a semantics for $\tptree{A}$ and its related terms boils down to giving a mathematical account of binary trees (in the same manner we have used the Kleene star to give a simple mathematical account of lists).
The semantics for terms involving these types can be found in the mechanization.

\begin{figure}[t!]
\hspace{-1.0em}
\begin{minipage}[t]{0.33\linewidth}
\begin{extsyntax}{Type}{A,B}{\; \cdots}
\sitem{\tpprod{A}{B}}{product}
\sitem{\tpstack{A}}{stack}
\sitem{\tptree{A}}{binary tree}
\end{extsyntax}
\end{minipage}
\hspace{2.5em}
\begin{minipage}[t]{0.6\linewidth}
\begin{extsyntax}{Term}{M,N}{\; \cdots}
\sitem{\tmrecord{M}{N}}{product intro}
\sitem{\tmproj{i}{M} \qquad\qquad\quad (i \in \{1,2\})}{product elim}
\sitem{\tmempty{}}{stack intro}
\sitem{\tmpush{M_h}{M_t}}{stack intro}
\sitem{\tmpop{M}{N_1}{x_h}{x_t}{N_2}}{stack elim}
\sitem{\tmleaf{}}{tree intro}
\sitem{\tmnode{M_d}{M_x}{M_l}{M_r}}{tree intro}
\sitem{\tmtrec{M}{N_1}{x_d}{x}{x_l}{x_r}{N_2}}{tree elim}
\end{extsyntax}
\end{minipage}
\caption{\lfplplus's type and term syntax.}
\label{fig:ext-syntax-tp-tm}
\end{figure}

\begin{figure}[t!]
\begin{mathpar}
\RuleTy{ProdI}{\types{\ctx}{\tmrecord{M_1}{M_2}}{\tpprod{A_1}{A_2}}}{\types{\ctx}{M_1}{A_1} \\ \types{\ctx}{M_2}{A_2}}
\and
\RuleTy{ProdE$_i$}{\types{\ctx}{\tmproj{i}{M}}{A_i}}{\types{\ctx}{M}{\tpprod{A_1}{A_2}}}
\and
\RuleTy{StackI$_1$}{\types{\ctx}{\tmempty}{\tpstack{A}}}{}
\and
\RuleTy{StackI$_2$}{\types{\ctxappend{\ctx_h}{\ctx_t}}{\tmpush{M_h}{M_t}}{\tpstack{A}}}{\types{\ctx_h}{M_h}{A} \\ \types{\ctx_t}{M_t}{\tpstack{A}}}
\and
\RuleTy{StackE}{\types{\ctxappend{\ctx}{\ctx'}}{\tmpop{M}{N_1}{x_h}{x_t}{N_2}}{B}}{\types{\ctx}{M}{\tpstack{A}} \\ \types{\ctx'}{N_1}{B} \\\\ \types{\ctxcons{\ctxcons{\ctx'}{x_h}{A}}{x_t}{\tpstack{A}}}{N_2}{B}}
\and
\RuleTy{TreeI$_1$}{\types{\ctx}{\tmleaf}{\tptree{A}}}{}
\and
\RuleTy{TreeI$_2$}{\types{\ctxappend{\ctx_d}{\ctxappend{\ctx_x}{\ctxappend{\ctx_l}{\ctx_r}}}}{\tmnode{M_d}{M_x}{M_l}{M_r}}{\tptree{A}}}{\types{\ctx_d}{M_d}{\tpdiam} \\ \types{\ctx_x}{M_x}{A} \\\\ \types{\ctx_l}{M_l}{\tptree{A}} \\ \types{\ctx_r}{M_r}{\tptree{A}}}
\and
\RuleTy{TreeE}{\types{\ctx}{\tmtrec{M}{N_1}{x_d}{x}{x_l}{x_r}{N_2}}{B}}{\types{\ctx}{M}{\tptree{A}} \\ \types{\ctxnil}{N_1}{B} \\\\ \types{\ctxcons{\ctxcons{\ctxcons{\ctxcons{\ctxnil}{x_d}{\tpdiam}}{x}{A}}{x_l}{B}}{x_r}{B}}{N_2}{B}}
\end{mathpar}
\caption{\lfplplus's typing rules.}
\label{fig:ext-jmt-ty}
\end{figure}

\enlargethispage{\baselineskip}
The stack type $\tpstack{A}$ is a particularly powerful addition to \lfpl. The vast majority of the difficulty in proving polynomial-time completeness of \lfpl, and by far the most complex aspect of our entire mechanization, was the definition and implementation of the bounded stack data structure discussed in \pref{sec:completeness-bounded-stack-definition} and \pref{sec:completeness-bounded-stack-implementation}. So, the addition of unbounded, $\tpdiam$-free stacks to \lfpl means \lfplplus admits a much simpler completeness proof, and thus it is worthwhile to show that they are sound despite their power. In essence, \lfplplus internalizes the bounded stack data structure from \pref{sec:completeness-bounded-stack-definition} and removes the boundedness restriction. No diamonds are required to interface with the stacks of \lfplplus, and there is no upper bound on the number of stack elements.

\subsection{Big-Step Operational Cost Semantics}
\label{sec:operational-semantics}

As discussed in \pref{sec:denotational-semantics}, any semantics for \lfplplus needs to have some way to inhabit the $\tpdiam$ type. To give \lfplplus an operational semantics, we introduce a separate language for \lfplplus values in \pref{fig:syntax-val}. Note that the definitions of value and environment are mutually recursive. This means we usually have to define functions on values via mutual recursion with a corresponding function on environments. Like in the denotational semantics, diamonds in lists and trees are not explicit.
By separating values from terms, we ensure \lfplplus terms remain polynomial-time sound while obtaining a means of providing concrete inputs to observe their evaluation. We also define a typing judgement $v : A$ for values, so that $\vdiam : \tpdiam$, $\vcons{v_h}{v_t} : \tplist{A}$ whenever $v_h : A$ and $v_t : \tplist{A}$, and so on.

\begin{figure}[t!]

\begin{minipage}[t]{0.45\linewidth}
\begin{syntax}{Val}{v}
\sitem{\vdiam}{diamond value}
\sitem{\vnull}{unit value}
\sitem{\vrecord{\env}{M_1}{M_2}}{lazy pair closure}
\sitem{\vinj{i}{v}}{injection}
\sitem{\vpair{v_1}{v_2}}{pair}
\sitem{\vlam{\env}{x}{M}}{function closure}
\sitem{\vempty}{empty stack}
\sitem{\vpush{v_h}{v_t}}{nonempty stack}
\sitem{\vnil}{empty list}
\sitem{\vcons{v_h}{v_t}}{nonempty list}
\sitem{\vleaf}{empty tree}
\sitem{\vnode{v}{v_l}{v_r}}{nonempty tree}
\end{syntax}
\end{minipage}
\hspace{3.5em}
\begin{minipage}[t]{0.45\linewidth}
\begin{syntax}{Env}{\env}
\sitem{\envnil}{empty envmt.}
\sitem{\envcons{\env}{x}{v}}{nonempty envmt.}
\end{syntax}
\end{minipage}
\caption{\lfplplus's value and environment syntax.}
\label{fig:syntax-val}
\end{figure}

\NewDocumentCommand{\RuleEval}{m m m}{\Rule{\text{Eval}:#1}{#2}{#3}\label{rule:eval-#1}}
\begin{figure}[t!]
\begin{mathpar}
\RuleEval{Var}{\evals{\envcons{\env}{x}{v}}{x}{\Cvar}{v}}{}
\and
\RuleEval{ArrowI}{\evals{\env}{\tmlam{x}{M}}{\Clam}{\vlam{\env}{x}{M}}}{}
\and
\RuleEval{SumE$_i$}{\evals{\envappend{\env}{\env'}}{\tmcase{M}{x}{N_1}{x}{N_2}}{c + c' + \Ccase}{v'}}{\evals{\env}{M}{c}{\vinj{i}{v}} \\ \evals{\envcons{\env'}{x}{v}}{N_i}{c'}{v'}}
\and
\RuleEval{ArrowE}{\evals{\envappend{\env_1}{\env_2}}{\tmapp{M}{N}}{c_1 + c_2 + c' + \Capp}{v'}}{\evals{\env_1}{M}{c_1}{\vlam{\env'}{x}{M'}} \\\\ \evals{\env_2}{N}{c_2}{v} \\ \evals{\envcons{\env'}{x}{v}}{M'}{c'}{v'}}
\and
\RuleEval{ListI$_2$}{\evals{\envappend{\env_d}{\envappend{\env_h}{\env_t}}}{\tmcons{M_d}{M_h}{M_t}}{c_d + c_h + c_t + \Ccons}{\vcons{v_h}{v_t}}}{\evals{\env_d}{M_d}{c_d}{\vdiam} \\ \evals{\env_h}{M_h}{c_h}{v_h} \\ \evals{\env_t}{M_t}{c_t}{v_t}}
\and
\RuleEval{TreeE$_2$}{\evals{\env}{\tmtrec{M}{N_1}{x_d}{x}{x_l}{x_r}{N_2}}{c + c_L + c_R + c' + \Ctrec}{v}}{\evals{\env}{M}{c}{\vnode{v_x}{v_l}{v_r}} \\ \evals{\envcons{\envnil}{y}{v_l}}{\tmtrec{y}{N_1}{x_d}{x}{x_l}{x_r}{N_2}}{c_L}{v_L} \\ \evals{\envcons{\envnil}{y}{v_r}}{\tmtrec{y}{N_1}{x_d}{x}{x_l}{x_r}{N_2}}{c_R}{v_R} \\ \evals{\envnil[x_d \mapsto \vdiam, x \mapsto v_x, x_l \mapsto v_L, x_r \mapsto v_R]}{N_2}{c'}{v}}
\end{mathpar}
\caption{Selected rules for \lfplplus's evaluation judgement.}
\label{fig:jmt-eval}
\end{figure}

To account for the cost of evaluation, we give a cost-annotated big-step operational semantics, represented by the judgement $\evals{\env}{M}{c}{v}$, read ``term $M$ evaluates under the environment $\env$ to a value $v$ with cost $c$''. A representative set of rules for this judgement is given in \pref{fig:jmt-eval}, and a complete set of rules can be found in{
  \ifdefined\FULLVERSION
    \pref{app:evaluation}%
  \else
    Appendix A of the full version~\cite{fullversion}%
  \fi
}. Note the presence of constants $\Cvar, \Cnull, \ldots, \Ctrec \in \mathbb N$ (one for each syntactic construct) in the conclusion of each rule. These are left arbitrary so that our cost model is generic with respect to the particular costs of each primitive operation of the language.

\begin{example}
Consider \code{reverse} from \pref{example:reverse}. Instead of proving correctness in terms of the denotational semantics, we could show that $\evals{\envcons{\envnil}{x}{v_\ell}}{\tmapp{\code{reverse}}{v_\ell}}{v_\ell'}$, where $v_\ell'$ represents the reverse of the list represented by $v_\ell$.
\end{example}

This semantics enjoys several standard properties such as determinism and preservation of typing, which we state here and prove in the mechanization.

\begin{theorem}[Determinism]
If $\evals{\env}{M}{c}{v}$ and $\evals{\env}{M}{c'}{v'}$ then $c = c'$ and $v = v'$.
\end{theorem}

\begin{theorem}[Preservation]
If $\types{\Gamma}{M}{A}$, $\envtypes{\env}{\Gamma}$, and $\evals{\env}{M}{c}{v}$, then $v : A$.
\end{theorem}
We can also give a denotational semantics for values and environments, and prove that they cohere with the denotational semantics for terms. In particular, given $v : A$ and $\envtypes{\env}{\Gamma}$, we define $\tpsem{v} \in \tpsem{A}$ and $\tpsem{\env} \in \tpsem{\Gamma}$ by mutual induction on the structures of $v$ and $\env$. Most cases are straightforward; for example, $\tpsem{\vpair{v_1}{v_2}} = (\tpsem{v_1}, \tpsem{v_2})$. For closures such as function values, we appeal to the term denotational semantics: $\tpsem{\vlam{\env}{x}{M}} = \lambda v . \tmsem{M}{\envcons{\tpsem{\env}}{x}{v}}$. We prove the following in the mechanization:

\begin{theorem}[Coherence]
\label{thm:coherence}
If $\types{\Gamma}{M}{A}$, $\envtypes{\env}{\Gamma}$, and $\evals{\env}{M}{c}{v}$, then $\tpsem{v} = \tmsem{M}{\tpsem{\env}}$.
\end{theorem}

Polynomial-time soundness is meant to be a statement about the computational complexity of functions encoded by \lfpl with respect to some machine model like a Turing machine. The big-step cost semantics is somewhat removed from a Turing machine's cost model.
However, it is well understood how to justify a cost semantics for a functional language like \lfplplus in terms of a low-level model such as a Turing machine~\cite{blellochProvableTimeSpace1996,Accattoli2019}.
The only non-standard features of \lfpl are the affine type system and the type $\tpdiam$, and both of these can be forgotten in a simple cost-preserving translation.
Thus, when proving true polynomial-time soundness of \lfpl, the only part which relies on \lfpl-specific reasoning is the very first step of giving a polynomial bound on the cost in \lfpl's big-step semantics; after that, these standard methods will carry the polynomial bound all the way down to the machine cost model. As the focus of this paper is \lfpl, we only mechanize the construction and verification of the polynomial cost bound.

\subsection{The Non-Size-Increasing Property}
\label{sec:soundness-size}
\enlargethispage{1.7\baselineskip}

We now make precise the non-size-increasing property referred to in previous sections. Our goal is to assign to each value $v$ a number $\valsize{v} \in \mathbb N$ that, roughly speaking, counts the number of diamond values $\vdiam$ contained within $v$.
We then show that the size of the value resulting from the evaluation of an \lfplplus term is bounded by the size of the values in the environment. We define size by mutual induction on values and environments in \pref{fig:size-val-env}.

\pagebreak
\begin{figure}[t!]
  \[ \substack{
  \begin{array}{r@{\;=\;}l@{\hspace{2.5em}}r@{\;=\;}l}
    \valsize{\vdiam} & 1 &
    \valsize{\vrecord{\env}{M_1}{M_2}} & \envsize{\env}  \\                           
    \valsize{\vnull} & 0 &
    \valsize{\vlam{\env}{x}{M}} &  \envsize{\env} \\                               
    \valsize{\vinj{i}{v}} & \valsize{v} &
    \valsize{\vpair{v_1}{v_2}} & \valsize{v_1} + \valsize{v_2} \\                                          
    \valsize{\vempty} & 0 &
    \valsize{\vpush{v_h}{v_t}} & \valsize{v_h} + \valsize{v_t} \\                            
    \valsize{\vnil} & 0 &
    \valsize{\vcons{v_h}{v_t}} & 1 + \valsize{v_h} + \valsize{v_t} \\                          
    \valsize{\vleaf} &  0 &
    \valsize{\vnode{v}{v_l}{v_r}} & 1 + \valsize{v} + \valsize{v_l} + \valsize{v_r}
  \end{array}\\[1ex]
  \begin{array}{r@{\;\;}c@{\;\;}lcr@{\;\;}c@{\;\;}l}  
   \envsize{\envnil} & = & 0 & &
  \envsize{\envcons{\env}{x}{v}} & = & \envsize{\env} + \valsize{v} 
  \end{array}}
  \]
\caption{Size of values and environments in \lfplplus.}
\label{fig:size-val-env}
\end{figure}

\begin{theorem}[The Non-Size-Increasing Property]
\label{thm:non-size-increasing}
Suppose that $\types{\ctx}{M}{A}$. Then, for all environments $\env$, values $v$, and costs $c$ such that $\evals{\env}{M}{c}{v}$, we have $\valsize{v} \le \envsize{\env}$.
\end{theorem}
\begin{proof}
By rule induction on $\evals{\env}{M}{c}{v}$. All cases are provided in the mechanization.
\end{proof}

The relevance of this result to soundness might already be clear. An immediate corollary is that there are no closed terms of type $\tpdiam$, because $\vdiam$ has size $1$ but the empty environment has size $0$. This is one of the main properties \lfpl's type system was designed to ensure.

\subsection{Polynomial Cost Bound}

Our goal for the rest of this section is to prove the following:

\begin{theorem}[Concrete Polynomial-Time Soundness of \lfplplus]
\label{thm:soundness-concrete}
Suppose $\types{\ctxcons{\ctxnil}{\ell}{\tplist{A}}}{M}{B}$, where $\dfree{A}$. Then, there exists a polynomial $P_M : \mathbb N \to \mathbb N$ such that, for any value $v_\ell : \tplist{A}$ representing a list of length $n$, there exists a value $v : B$ and a cost $c \in \mathbb N$ such that $\evals{\envcons{\envnil}{\ell}{v_\ell}}{M}{c}{v}$ and $c \le P(n)$.
\end{theorem}

Note that it is important that the polynomial only depends on $M$, i.e., it is quantified before the input $v_\ell$. Taking $B = \tplist{\tpunit}$ and $A = \tpunit$, we recover the basic notion of soundness for \lfplplus functions on natural numbers. We present this more general form to motivate our transition to the even more general form we need to successfully prove it. As usual, we must reason about arbitrary terms under arbitrary contexts. It is tempting to say the following:
\begin{center}
\parbox{.9\linewidth}{
  \it
Suppose $\types{\ctx}{M}{A}$. There exists a polynomial $P_M : \mathbb N \to \mathbb N$ such that, if $\env : \Gamma$, there exists a value $v : A$ and a cost $c \in \mathbb N$ such that $\evals{\env}{M}{c}{v}$ and $c \le P(\envsize{\env})$.
}
\end{center}
However, this statement is false. The problem is that the environment may contain values which incur cost when used, which is dynamic information and $P_M$ only has access to static information. Consider $\types{f : \tparr{\tpunit}{\tpunit}}{\tmapp{f}{\tmnull}}{\tpunit}$. If for simplicity $\Capp = \Cvar = \Cnull = 0$, we cannot statically observe how this program could incur any cost. After all, the only statically observable syntactic constructs are application, variable, and null. On the other hand, dynamically, $\env$ could map $f$ to a function which incurs cost.

To resolve this issue, we must consider the potential cost of running functions (or lazy pairs) stored in the environment $\env$ and result value $v$. Suppose for a moment we have defined three families of polynomials: $P_M$ on all terms $M$, $P_v$ on all values $v$, and $P_\env$ on all environments $\env$. Our new goal is to prove the following theorem.

\begin{theorem}[General Soundness]
\label{thm:soundness-general}
Suppose $\types{\ctx}{M}{A}$. Then, for any $\env : \Gamma$, there exists a value $v : A$ and a cost $c \in \mathbb N$ such that $\evals{\env}{M}{c}{v}$ and:
\[ c + P_v(\envsize{\env}) \le P_M(\envsize{\env}) + P_{\env}(\envsize{\env}) \]
\end{theorem}

This theorem is true (for some definitions of $P_v$, $P_M$, and $P_\env$), though as stated it is not strong enough to be proven by rule induction on $\types{\ctx}{M}{A}$. Before discussing this further, we define the polynomials used in this theorem. By $\max(P,Q) : \mathbb N \to \mathbb N$ we mean the polynomial with coefficients equal to the pairwise maximums of the coefficients of polynomials $P,Q : \mathbb N \to \mathbb N$. This way, $\max(P(n), Q(n)) \le \max(P,Q)(n)$.

\begin{definition}[Term Polynomial]
Let $M$ be a \lfplplus term. We inductively define the polynomial $P_M : \mathbb N \to \mathbb N$ as follows.
{\upshape
\begin{align*}
P_{x}(n) &= \Cvar \\
P_{\tmrecord{M_1}{M_2}}(n) &= \Crecord + \max(P_{M_1}, P_{M_2})(n) \\
P_{\tmproj{i}{M}}(n) &= \Cproj{i} + P_M(n) \\
P_{\tmlam{x}{M}}(n) &= \Clam + P_M(n) \\
P_{\tmapp{M}{N}}(n) &= \Capp + P_M(n) + P_N(n) \\
P_{\tmcase{M}{x_1}{N_1}{x_2}{N_2}}(n) &= \Ccase + P_M(n) + \max(P_{N_1}, P_{N_2})(n) \\
P_{\tmrec{M}{N_1}{x_d}{x_h}{x_t}{N_2}}(n) &= P_M(n) + (\Crec + P_{N_1}(n)) 
 + n (\Cvar + \Crec + P_{N_2}(n)) \\
P_{\tmtrec{M}{N_1}{x}{y}{z}{w}{N_2}}(n) &= P_M(n) {+} (n {+} 1) (\Ctrec {+} P_{N_1}(n)) 
 {+} n (2\Cvar {+} \Ctrec {+} P_{N_2}(n))
\end{align*}
}
The remaining cases are similar and provided in the mechanization.
\end{definition}

The aim of $P_M$ is to bound the cost of evaluating $M$, and the future cost of using any closures it evaluates to, by $P_M(n)$ in an environment containing $n$ diamonds. The intuition behind $P_{\tmrec{M}{N_1}{x_d}{x_h}{x_t}{N_2}}(n)$ is that $M$ and $N_1$ are evaluated once, and $N_2$ could be evaluated up to $n$ times, where $n = \envsize{\env}$ is the number of diamonds available in the environment. Similarly, for $P_{\tmtrec{M}{N_1}{x}{y}{z}{w}{N_2}}$, the base case $N_1$ could be evaluated $n + 1$ times (a tree with $n$ internal nodes has $n + 1$ leaf nodes) and $N_2$ could be evaluated $n$ times.

\begin{definition}[Value and Environment Polynomials]
We define polynomials $P_\env : \mathbb N \to \mathbb N$ and $P_v : \mathbb N \to \mathbb N$ via mutual induction as follows:
\[
\begin{array}{r@{\;=\;}l@{\hspace{4.5em}}r@{\;=\;}l}
  P_{\vdiam} = P_{\vnull} & P_{\vempty} = P_{\vnil} = P_{\vleaf} = 0 &
  P_{\envnil} & 0 \\
 P_{\vrecord{\env}{M_1}{M_2}} & P_{\env} + \max(P_{M_1}, P_{M_2}) &
  P_{\envcons{\env}{x}{v}} & P_\env + P_v \\
P_{\vlam{\env}{x}{M}} & P_{\env} + P_M \\
P_{\vpair{v_1}{v_2}} & P_{v_1} + P_{v_2} 
\end{array}
\]
The definition of $P_v$ for other values $v$ is given in the mechanization.
\end{definition}
The idea of  $P_v$ and $P_\env$ is to bound the future cost of evaluating all of the closures (functions and lazy pairs) stored inside $v$ and $\env$ respectively. With these choices of polynomials, \pref{thm:soundness-general} implies \pref{thm:soundness-concrete}, because when $v : \tplist{A}$ and $\dfree{A}$, it is clear that $P_v = 0$ and $\valsize{v} = n$, where $n$ is the length of the list represented by $v$.

\begin{example}
Consider \code{reverse} and \code{revAppend} from \pref{example:reverse}. For simplicity, assume our cost model has $\Crec = \Capp = 1$ and all other constants are zero. We have:
\[ P_\code{\tmapp{reverse}{\ell}}(n) = 3\Capp + P_\code{revAppend}(n) = 3\Capp + \Crec + n(\Crec + \Capp) = 2n + 4 \]
Applying \pref{thm:soundness-concrete} to this result, we learn that if \code{reverse} is evaluated with an input list of length $n$ (and an $\tpdiam$-free element type), we will incur at most $2n + 4$ cost, thus proving \code{reverse} is linear-time.
\end{example}

\subsection{Proving the Soundness Theorem}
\label{sec:proving-general-soundness}

To prove \pref{thm:soundness-general}, we take a standard logical relations approach. We define a unary logical relation $\mathcal R$ for terms, values, and environments with the goal that the fundamental theorem for $\mathcal R$ should imply \pref{thm:soundness-general}.

\begin{definition}
We first define the shorthand $\Rterm{A}{\env}{M}$ to mean:
\begin{align*}
&\text{For all $n \ge \envsize{\env}$, there exists $v : A$ and a cost $c \in \mathbb N$ such that $\evals{\env}{M}{c}{v}$,} \\
&\text{$\Rval{A}{v}$ holds, and $c + P_v(n) \le P_M(n) + P_\env(n)$}
\end{align*}
\end{definition}
By working with a general $n \ge \envsize{\env}$ rather than just $\envsize{\env}$ as the input to the polynomials, we can deal with the case where $\env$ arises as part of a larger environment during the proof. This is reminiscent of \emph{future worlds} from Kripke-style logical relations~\cite{Mitchell1991}.

\begin{definition}
We define $\Rval{A}{v}$ by induction on the type $A$:
\begin{align*}
\Rval{\tpunit}{v} &\quad \text{iff} \quad v = \vnull \\
\Rval{\tptens{A_1}{A_2}}{v} &\quad \text{iff} \quad v = \vpair{v_1}{v_2}, \; \Rval{A_1}{v_1}, \text{ and } \Rval{A_2}{v_2} \\
\Rval{\tpprod{A_1}{A_2}}{v} &\quad \text{iff} \quad v = \vrecord{\env}{M_1}{M_2}, \; \Rterm{A_1}{\env}{M_1}, \text{ and } \Rterm{A_2}{\env}{M_2} \\
\Rval{\tparr{A}{B}}{v} &\quad \text{iff} \quad v = \vlam{\env}{x}{M}, \text{ and } \Rval{A}{v} \text{ implies } \Rterm{B}{\envcons{\env}{x}{v}}{M} \\
\Rval{\tplist{A}}{v} &\quad \text{iff} \quad \text{$v$ is a list $\ell$ such that } \Rval{A}{-} \text{ holds for each element of $\ell$ }
\end{align*}
The remaining cases are provided in the mechanization.
\end{definition}

\begin{definition}
We define $\Renv{\Gamma}{\env}$ by induction on $\Gamma$. $\Renv{\ctxnil}{\envnil}$ always holds. $\Renv{\ctxcons{\Gamma}{x}{A}}{\envcons{\env}{x}{v}}$ holds whenever $\Rval{A}{v}$ and $\Renv{\Gamma}{\env}$ hold.
\end{definition}

Now, \pref{thm:soundness-general} is clearly implied by the fundamental theorem for our relation $\mathcal R$:

\begin{theorem}[Fundamental Theorem]
Suppose $\types{\Gamma}{M}{A}$ and $\envtypes{\env}{\Gamma}$. Then, if $\Renv{\Gamma}{\env}$ holds, so does $\Rterm{A}{\env}{M}$.
\end{theorem}
\begin{proof}
We go by rule induction on $\types{\ctx}{M}{A}$. \pref{thm:non-size-increasing} comes up in nearly every case, since we need to satisfy the premise $n \ge \envsize{\env}$ of $\Rterm{A}{\env}{M}$ to use any inductive hypotheses. Importantly, we actually make use of this premise in the list and tree recursor cases, where we can say for example that the length of any list value $v$ is at most $\envsize{\env}$ and therefore at most $n$. Most cases are verbose but straightforward, and a couple of the less straightforward cases can be found in{
  \ifdefined\FULLVERSION
    \pref{app:soundness}%
  \else
    Appendix B of the full version~\cite{fullversion}%
  \fi
}.
\end{proof}

\section{Mechanization}
\label{sec:mechanization}

We have mechanized all results stated in Sections \ref{sec:lfpl}, \ref{sec:completeness}, and \ref{sec:soundness} in the proof assistant Istari~\cite{istariProofAssistant}. To our knowledge, our mechanization serves as one of two published case studies in this relatively new proof assistant~\cite{Li2026}. The mechanization is roughly 7,000 lines of meaningful Istari code. Around 3,000 of these are definitions and \lfpl code, around 3,500 are proof scripts, and around 800 are lemma and theorem statements~\cite{mechanization}.

\subsection{Defining LFPL}
\enlargethispage{\baselineskip}

We begin by defining the types of \lfpl as an inductive datatype \lstinline{tp} within Istari. \lfpl does not have type variables or similar complexities. We also define the judgement $\dfree{A}$ within Istari. Next, we represent contexts $\ctx$ as an Istari list of \lfpl types. We do not need variable names associated with the types because we use De Bruijn indices in the term structure. Importantly, we define a judgement $\code{split} \; \ctx \; \ctx_1 \; \ctx_2$ to indicate when $\ctx$ can be  bipartitioned into $\ctxappend{\ctx_1}{\ctx_2}$, used extensively in the typing rules. Essentially, viewing the contexts as type lists as we do in our mechanization, our definition of $\code{split} \; \ctx \; \ctx_1 \; \ctx_2$ ensures $\code{append}(\ctx_1)(\ctx_2)$ is equivalent to some permutation of $\ctx$.

Instead of presenting an untyped syntax with an extrinsic typing judgement, as we do in \pref{sec:lfpl-definition}, our mechanization defines \lfpl's term syntax and typing rules at once, so that all terms are intrinsically well-typed under a known context. These intrinsically typed terms are represented by an inductive datatype \lstinline{term : ctx -> tp -> type}, where \lstinline{type} refers to Istari types, so that a value of type \lstinline{term G A} corresponds to a term $M$ such that $\types{G}{M}{A}$. For example, the rule for the term $\tmpair{M_1}{M_2}$ is:
\begin{lstlisting}
| pair : forall (G GA GB : ctx) (A B : tp) .
    split G GA GB -> term GA A -> term GB B -> term G (tensor A B)
\end{lstlisting}
It says we must bipartition the context $G$ into $\ctxappend{G_A}{G_B}$, and give a well-typed term under each of those contexts to form a pair $\tmpair{M_1}{M_2}$ which has type $\tptens{A}{B}$ under context $G$.

\subsection{Denotational and Operational Semantics}

Instead of the set-theoretic denotational semantics we give in \pref{sec:denotational-semantics}, our mechanization gives a denotational semantics from intrinsically typed \lfpl terms into Istari terms of the appropriate type. This is similar to the set-theoretic formulation, since all operations we use on sets (products, coproducts, etc.) are also present in the Istari type system. One benefit is that we can rely on Istari's reduction engine to simplify Istari terms resulting from the semantics of a complicated \lfpl term. These denotations have the following signatures.
Here, a value of type \lstinline{env G} represents an environment $\env$ such that $\env : G$. 
\begin{lstlisting}
tp_sem : tp -> type
term_sem : forall (G : ctx) (A : tp) . term G A -> env G -> tp_sem A
\end{lstlisting}
Similarly, for the operational semantics judgement
$\evals{\env}{M}{c}{v}$ from \pref{sec:operational-semantics}, we
define an inductive datatype encoding its rules with the following
type:
\begin{lstlisting}
evals : forall (G : ctx) (A : tp) . term G A -> env G -> value A -> nat -> type
\end{lstlisting}
As briefly mentioned at the end of \pref{sec:operational-semantics}, we give a denotational semantics for values and prove a coherence result (\pref{thm:coherence}) in our mechanization, which is a good example of what many of our Istari theorems look like.
\begin{lstlisting}
value_sem : forall (A : tp) . value A -> tp_sem A
value_sem = ...
lemma "operational_equiv_denotational"
  forall (G : ctx) (A : tp) (M : term G A) (V : env G) (v : value A) (c : nat) .
  evals M V v c -> value_sem v = term_sem M (env_sem V)
\end{lstlisting}

\subsection{Takeaways from Istari}

When beginning this project, we were uncertain which proof assistant would be most suitable for the mechanization, so we began by developing two mechanizations in parallel, one in Agda and one in Istari. During preliminary work for the completeness theorem, we ran into two significant difficulties with Agda.

The first issue was a lack of tactics. Agda expressions are often cleaner and more readable than the proof scripts of tactic-based proof assistants like Rocq and Istari. However, when one constantly needs to reason about equality, say by substituting $e'$ for $e$ in a term $f(e)$ where $f$ is a complex expression, it becomes helpful to have a substitution tactic. In Agda, we needed to manually specify $f$, which is infeasible when $f$ is large.

Secondly, Agda requires that all expressions are well-typed, even in the statement of a theorem yet to be proven. This presents some difficulties. For example, consider the \lfpl type $(\tplist{\tpunit})^k$ used in \pref{sec:completeness-bounded-stack-definition}. We know that $\tpsem{(\tplist{\tpunit})^k} = \tpsem{\tplist{\tpunit}}^k$, so we implicitly switch between these forms as needed, but in the mechanization this must be explicit. Some parts of our proof are only well-typed if they are given the form on the left-hand side of the equality, and others need the right-hand form. In Agda, we had to state a theorem where some parts of the theorem needed the left form, and others needed the right, so there was no way to formulate this such that it was well-typed without using transport functions, which are known to cause other major difficulties~\cite{tabareau2018}. Istari does not require the well-typedness of a theorem before it is proven. It must be given a type by the end of the proof, but this is not difficult since we have a lemma stating $\tpsem{(\tplist{\tpunit})^k} = \tpsem{\tplist{\tpunit}}^k$. The difference is that Istari lets us show well-typedness during the proof, whereas Agda demands it before the proof.

Such issues led us to finish the completeness and soundness proofs only in Istari. Retrospectively, we can say that it was a more than adequate tool for this task. In fact, we never encountered any Istari-related roadblocks. This is not to say there weren't annoyances; as a relatively new proof assistant, Istari sometimes struggles to infer implicit arguments and automatically eliminate vacuous proof cases, requiring manual intervention in situations where Agda would succeed.

\section{Related Work}

The most closely related work to ours is that of Atkey~\cite{atkeyPolynomialTimeDependent2024}.
He develops a powerful extension of \lfpl with full dependent types
that includes lists that do not enjoy iteration, much
like the stack types we introduced in \pref{sec:soundness}.
This separation of (unrestricted) construction of data-structures and iteration has
been previously used by Jones~\cite{Jones2001} in his
``Cons-free'' system for characterizing the class P.
Atkey also provides a semantic polynomial-time soundness proof using the realizability method of Dal Lago and Hofmann~\cite{dallagoRealizabilityModelsImplicit2011}. These developments are mechanized in Agda. Our method for proving soundness differs from his in that it is syntactic rather than semantic and constructs an explicit polynomial bound on the cost of evaluation for each term. Furthermore, we devote great attention to mechanizing polynomial-time completeness, ensuring the strongest result by working with a minimal version of \lfpl.

Hofmann's original work on LFPL~\cite{hofmannLinearTypesNonsizeincreasing1999,hofmannStrengthNonsizeIncreasing2002} is the basis of this article and has been discussed in detail in Sections~\ref{sec:lfpl},~\ref{sec:completeness}, and~\ref{sec:soundness}.
The soundness proof is based on previous work by Aehlig and Schwichtenberg~\cite{aehligSyntacticalAnalysisNonsizeincreasing2002}, which introduces the technique of constructing explicit polynomials from \lfpl syntax to bound the cost of evaluation. Ours is the first work to apply their technique to a big-step operational semantics.
None of these articles comes with a mechanization or full account of both soundness and completeness.

\lfpl also directly inspired automatic amortized resource analysis (AARA)~\cite{HoffmJ22}, and so our work could perhaps inspire early steps towards mechanizing the metatheory of AARA. Additionally, there is a polynomial-time completeness result for AARA similar to the one for \lfpl discussed in \pref{sec:completeness}~\cite{phamTypableFragmentsPolynomial2021}.

\enlargethispage{\baselineskip}
Beyond \lfpl, there is a large body of work on implicit computational complexity (ICC)~\cite{dallagoImplicitComputationComplexity2022}. The polynomial time and space classes are a common focus~\cite{gaboardiSoftLinearLogic2008,gaboardiImplicitCharacterizationPSPACE2012}, though there has been interest in logarithmic space usage as well~\cite{schoppStratifiedBoundedAffine2007, ramyaaRamifiedCorecurrenceLogspace2011}. Interestingly, Hofmann~\cite{hofmannStrengthNonsizeIncreasing2002} noticed that, if the structural recursor of \lfpl is replaced with general recursion, the terminating fragment of the language characterizes exponential-time computation. \lfpl is an influential language in ICC, helping to inspire numerous developments in the field~\cite{dallagoRealizabilityModelsImplicit2011,marion2003,marion2011,baillot2016,dalLago2009,bonfante2011}.

We are only aware of two other works in ICC, besides ours and
Atkey's~\cite{atkeyPolynomialTimeDependent2024}, that are
significantly mechanized.
Heraud and Nowak~\cite{Heraud2011} mechanized the soundness and
completeness of the language for P of Bellantoni and
Cook~\cite{Bellantoni92} in Rocq.
This language is quite different from LFPL since it does not feature a
linear type discipline or higher-order functions but instead relies on
the idea of \emph{safe arguments}.
Férée et al.~\cite{Feree2018} describe a library for verifying
polynomial-time complexity in Rocq by using
quasi-interpretations. Instead of a higher-order language with a
linear type system, they consider first-order term-rewrite systems. As
a result, the soundness and completeness arguments are very different
from the ones for \lfpl.
In particular, the completeness proofs for both, the language of
Bellantoni and Cook and the system of quasi-interpretations, do not
require the construction of bounded stacks, which is the main
difficulty in completeness proof for \lfpl.

\section{Conclusion}

This article provides a self-contained presentation of Hofmann's \lfpl.
The mechanized soundness proof constructs, for each \lfpl expression,
an explicit polynomial that bounds the evaluation cost with respect to
a big-step cost semantics.
The mechanized completeness proof shows how \lfpl can simulate
polynomial-time Turing machines, relying only on linear lists and
functions to encode the polynomially sized tape in a non-size-increasing
way.
Both proofs are based on existing ideas but contain significant
simplifications and improvements.

Our hope is that this work contributes to making \lfpl's metatheory
more accessible and to inspiring \lfpl-based innovations in areas like
automatic resource bound analysis~\cite{HoffmJ22}, quantitative type
theory~\cite{atkeyPolynomialTimeDependent2024}, and memory management
in functional languages~\cite{lorenzenFP2FullyInPlace2023}.

\section*{Acknowledgments}

This material is based upon work supported by Jane Street Group, LLC
and the {National Science Foundation under grants
no.~2311983 and 2525102.
Any opinions, findings, and conclusions or recommendations in this
material are those of the authors and do not necessarily reflect the
views of the NSF.


\bibliography{references}

\ifdefined\FULLVERSION
  \newpage
  \appendix
  \section{Evaluation Judgement Rules}
\label{app:evaluation}

In \pref{fig:app-jmt-eval-1} and \pref{fig:app-jmt-eval-2}, we provide all rules for the big-step evaluation judgement, $\evals{\env}{M}{c}{v}$, defined in \pref{sec:operational-semantics}. The additional language features of \lfplplus are included.

  \RenewDocumentCommand{\RuleEval}{m m m}{\Rule{\text{Eval}:#1}{#2}{#3}}
\begin{figure*}[t!]
\begin{mathpar}
\RuleEval{Var}{\evals{\env}{x}{\Cvar}{\env(x)}}{}
\and
\RuleEval{UnitI}{\evals{\env}{\tmnull}{\Cnull}{\vnull}}{}
\and
\RuleEval{SumE$_i$}{\evals{\envappend{\env}{\env'}}{\tmcase{M}{x}{N_1}{x}{N_2}}{c + c' + \Ccase}{v'}}{\evals{\env}{M}{c}{\vinj{i}{v}} \\ \evals{\envcons{\env'}{x}{v}}{N_i}{c'}{v'}}
\and
\RuleEval{SumI$_i$}{\evals{\env}{\tminj{i}{M}}{c + \Cinj{i}}{\vinj{i}{v}}}{\evals{\env}{M}{c}{v}}
\and
\RuleEval{TensorI}{\evals{\envappend{\env_1}{\env_2}}{\tmpair{M_1}{M_2}}{c_1 + c_2 + \Cpair}{\vpair{v_1}{v_2}}}{\evals{\env_1}{M_1}{c_1}{v_1} \\ \evals{\env_2}{M_2}{c_2}{v_2}}
\and
\RuleEval{TensorE}{\evals{\envappend{\env}{\env'}}{\tmletp{M}{x_1}{x_2}{N}}{c + c' + \Cletp}{v}}{\evals{\env}{M}{c}{\vpair{v_1}{v_2}} \\\\ \evals{\env'[x_1 \mapsto v_1, x_2 \mapsto v_2]}{N}{c'}{v}}
\and
\RuleEval{ArrowE}{\evals{\envappend{\env_1}{\env_2}}{\tmapp{M}{N}}{c_1 + c_2 + c' + \Capp}{v'}}{\evals{\env_1}{M}{c_1}{\vlam{\env'}{x}{M'}} \\\\ \evals{\env_2}{N}{c_2}{v} \\ \evals{\envcons{\env'}{x}{v}}{M'}{c'}{v'}}
\and
\RuleEval{ArrowI}{\evals{\env}{\tmlam{x}{M}}{\Clam}{\vlam{\env}{x}{M}}}{}
\and
\RuleEval{ListE$_1$}{\evals{\envappend{\env}{\env'}}{\tmrec{M}{N_1}{x_d}{x_h}{x_t}{N_2}}{c + c' + \Crec}{v}}{\evals{\env}{M}{c}{\vnil} \\ \evals{\env'}{N_1}{c'}{v}}
\and
\RuleEval{ListI$_2$}{\evals{\envappend{\env_d}{\envappend{\env_h}{\env_t}}}{\tmcons{M_d}{M_h}{M_t}}{c_d + c_h + c_t + \Ccons}{\vcons{v_h}{v_t}}}{\evals{\env_d}{M_d}{c_d}{\vdiam} \\ \evals{\env_h}{M_h}{c_h}{v_h} \\ \evals{\env_t}{M_t}{c_t}{v_t}}
\and
\RuleEval{ListI$_1$}{\evals{\env}{\tmnil}{\Cnil}{\vnil}}{}
\and
\RuleEval{ListE$_2$}{\evals{\envappend{\env}{\env'}}{\tmrec{M}{N_1}{x_d}{x_h}{x_t}{N_2}}{c + c_T + c' + \Crec}{v_2}}{\evals{\env}{M}{c}{\vcons{v_h}{v_t}} \\ \evals{\envcons{\env'}{y}{v_t}}{\tmrec{y}{N_1}{x_d}{x_h}{x_t}{N_2}}{c_T}{v_T} \\ \evals{\cdot[x_d \mapsto \vdiam, x_h \mapsto v_h, x_t \mapsto v_T]}{N_2}{c'}{v_2}}
\end{mathpar}
\caption{The full definition of \lfplplus's evaluation judgement. (1/2)}
\label{fig:app-jmt-eval-1}
\end{figure*}
\begin{figure*}[t!]
\begin{mathpar}
\RuleEval{ProdI}{\evals{\env}{\tmrecord{M_1}{M_2}}{\Crecord}{\vrecord{\env}{M_1}{M_2}}}{}
\and
\RuleEval{ProdE$_i$}{\evals{\env}{\tmproj{i}{M}}{c + c' + \Cproj{i}}{v}}{\evals{\env}{M}{c}{\vrecord{\env'}{M_1}{M_2}} \\\\ \evals{\env'}{M_i}{c'}{v}}
\and
\RuleEval{StackI$_1$}{\evals{\env}{\tmempty}{\Cempty}{\vempty}}{}
\and
\RuleEval{StackI$_2$}{\evals{\envappend{\env_h}{\env_t}}{\tmpush{M_h}{M_t}}{c_h {+} c_t {+} \Cpush}{\vpush{v_h}{v_t}}}{\evals{\env_h}{M_h}{c_h}{v_h} \\ \evals{\env_t}{M_t}{c_t}{v_t}}
\and
\RuleEval{StackE$_1$}{\evals{\envappend{\env}{\env'}}{\tmpop{M}{N_1}{x_h}{x_t}{N_2}}{c {+} c' {+} \Cpop}{v}}{\evals{\env}{M}{c}{\vempty} \\ \evals{\env'}{N_1}{c'}{v}}
\and
\RuleEval{StackE$_2$}{\evals{\envappend{\env}{\env'}}{\tmpop{M}{N_1}{x_h}{x_t}{N_2}}{c + c' + \Cpop}{v}}{\evals{\env}{M}{c}{\vpush{v_h}{v_t}} \\\\ \evals{\env'[x_h \mapsto v_h, x_t \mapsto v_t]}{N_2}{c'}{v}}
\and
\RuleEval{TreeI$_2$}{\evals{\envappend{\env_d}{\envappend{\env_x}{\envappend{\env_l}{\env_r}}}}{\tmnode{M_d}{M_x}{M_l}{M_r}}{c_d + c_x + c_l + c_r + \Cnode}{\vnode{v_x}{v_l}{v_r}}}{\evals{\env_d}{M_d}{c_d}{\vdiam} \\ \evals{\env_x}{M_x}{c_x}{v_x} \\ \evals{\env_l}{M_l}{c_l}{v_l} \\ \evals{\env_r}{M_r}{c_r}{v_r}}
\and
\RuleEval{TreeI$_1$}{\evals{\env}{\tmleaf}{\Cleaf}{\vleaf}}{}
\and
\RuleEval{TreeE$_1$}{\evals{\env}{\tmtrec{M}{N_1}{x_d}{x}{x_l}{x_r}{N_2}}{c + c' + \Ctrec}{v}}{\evals{\env}{M}{c}{\vleaf} \\ \evals{\envnil}{N_1}{c'}{v}}
\and
\RuleEval{TreeE$_2$}{\evals{\env}{\tmtrec{M}{N_1}{x_d}{x}{x_l}{x_r}{N_2}}{c + c_L + c_R + c' + \Ctrec}{v}}{\evals{\env}{M}{c}{\vnode{v_x}{v_l}{v_r}} \\ \evals{\envcons{\envnil}{y}{v_l}}{\tmtrec{y}{N_1}{x_d}{x}{x_l}{x_r}{N_2}}{c_L}{v_L} \\ \evals{\envcons{\envnil}{y}{v_r}}{\tmtrec{y}{N_1}{x_d}{x}{x_l}{x_r}{N_2}}{c_R}{v_R} \\ \evals{\envnil[x_d \mapsto \vdiam, x \mapsto v_x, x_l \mapsto v_L, x_r \mapsto v_R]}{N_2}{c'}{v}}
\end{mathpar}
\caption{The full definition of \lfplplus's evaluation judgement. (2/2)}
\label{fig:app-jmt-eval-2}
\end{figure*}

  \section{Soundness Proof Cases}
\label{app:soundness}

Here, we provide two cases of the soundness proof. The other cases are similar to these, though often more straightforward, and can all be found in the mechanization~\cite{mechanization}. First, we go through the function application case in moderate detail to demonstrate the need for a logical relation. Second, we sketch the list recursor case, which as a driver of variable-time computation is one of the most critical points to consider for polynomial-time soundness.
\begin{theorem}[Fundamental Theorem]
Suppose $\types{\Gamma}{M}{A}$ and $\envtypes{\env}{\Gamma}$. Then, if $\Renv{\Gamma}{\env}$ holds, so does $\Rterm{A}{\env}{M}$.
\end{theorem}
\begin{proof}
We go by rule induction on $\types{\Gamma}{M}{A}$.

\paragraph{Application Case}
Suppose $\types{\ctxappend{\ctx_1}{\ctx_2}}{\tmapp{M}{N}}{B}$ because $\types{\ctx_1}{M}{\tparr{A}{B}}$ and $\types{\ctx_2}{N}{A}$. Assuming $\Renv{\Gamma}{\env}$, our goal is to show that $\Rterm{B}{\env}{\tmapp{M}{N}}$.

Let $\envappend{\env_1}{\env_2}$ be the bipartition of $\env$ corresponding to $\ctxappend{\ctx_1}{\ctx_2}$. Unpacking the definition of $\Rterm{B}{\env}{\tmapp{M}{N}}$, we need to show that, for all $n \ge \envsize{\env_1} + \envsize{\env_2}$, there is a value $v : B$ and a cost $c \in \mathbb N$ such that $\Rval{B}{v}$ holds, $\evals{\env}{\tmapp{M}{N}}{c}{v}$, and:
\[ c + P_v(n) \le \Capp + P_M(n) + P_N(n) + P_{\env_1}(n) + P_{\env_2}(n) \]

Since $\Renv{\Gamma}{\env}$, we have $\Renv{\Gamma_1}{\env_1}$ and $\Renv{\Gamma_2}{\env_2}$. Therefore, by the induction hypotheses, we obtain $\Rterm{\tparr{A}{B}}{\env_1}{M}$ and $\Rterm{A}{\env_2}{N}$. Applying each to $n$ (which is valid since $n \ge \envsize{\env} \ge \envsize{\env_i}$ for $i \in \{1,2\}$), we conclude that:
\begin{itemize}
  \item There is a value $v_M$ and a cost $c_M$ such that $\Rval{\tparr{A}{B}}{v}$ holds, $\evals{\env_1}{M}{c_M}{v_M}$, and:
  \[ c_M + P_{v_M}(n) \le P_M(n) + P_{\env_1}(n) \]
  \item There is a value $v_N$ and a cost $c_N$ such that $\Rval{B}{v_N}$ holds, $\evals{\env_2}{N}{c_N}{v_N}$, and:
  \[ c_N + P_{v_N}(n) \le P_N(n) + P_{\env_2}(n) \]
\end{itemize}

By definition of $\Rval{\tparr{A}{B}}{v}$, we know that $v_M = \vlam{\env'}{x}{M'}$ and $\Rterm{B}{\envcons{\env'}{x}{v_N}}{M'}$. In particular, by definition of $P_{\vlam{\env'}{x}{M'}}$, we obtain the following inequality:
\[ c_M + P_{\env'}(n) + P_{M'}(n) \le P_M(n) + P_{\env_1}(n) \]
Also, by \pref{thm:non-size-increasing}, we know:
\begin{align*}
\envsize{\envcons{\env'}{x}{v_N}} &= \envsize{\env'} + \valsize{v_N} = \envsize{v_M} + \valsize{v_N} \\
&\le \envsize{\env_1} + \envsize{\env_2} \le n
\end{align*}
Therefore, we may apply $\Rterm{B}{\envcons{\env'}{x}{v_N}}{M'}$ to $n$, concluding that there is a value $v : B$ and a cost $c_{M'}$ such that $\Rval{B}{v}$ holds, $\evals{\envcons{\env'}{x}{v_N}}{M'}{c}{v}$, and:
\[ c_{M'} + P_v(n) \le P_{M'}(n) + P_{v_N}(n) + P_{\env'}(n) \]
Take $c = c_M + c_N + c_{M'} + \Capp$. After applying \pref{rule:eval-ArrowE} to conclude the first part of our goal, it remains to prove the cost bound. Chaining all our inequalities together, we get that:
\[ c_M + c_N + c_{M'} + P_v(n) \le P_M(n) + P_N(n) + P_{\env}(n) \]
Adding $\Capp$ to both sides yields the desired inequality.

\paragraph{Recursor Case} Suppose $\types{\ctxappend{\ctx_1}{\ctx_2}}{\tmrec{M}{N_1}{x_d}{x_h}{x_t}{N_2}}{B}$ because:
\begin{itemize}
  \item $\types{\ctx_1}{M}{\tplist{A}}$
  \item $\types{\ctx_2}{N_1}{B}$
  \item $\types{\ctxcons{\ctxcons{\ctxcons{\ctxnil}{x_d}{\tpdiam}}{x_h}{A}}{x_t}{B}}{N_2}{B}$
\end{itemize}

Assume $\Renv{\Gamma_1}{\env_1}$ and $\Renv{\Gamma_2}{\env_2}$, where $\envappend{\env_1}{\env_2}$ is the bipartition corresponding to $\ctxappend{\Gamma_1}{\Gamma_2}$. By the inductive hypothesis for $M$, we know that there is a value $v_M$ and cost $c_M$ such that $\Rval{\tplist{A}}{v_M}$ holds, $\evals{\env_1}{M}{c_M}{v_M}$, and:
\[ c_M + P_{v_M}(n) \le P_M(n) + P_{\env_1}(n) \]
By definition of $\Rval{\tplist{A}}{v_M}$, we know that $v_M$ is a list $\ell$ of values of type $A$ for which $\Rval{A}{-}$ holds. Now, recall the definition of $P_{\tmrec{M}{N_1}{x_d}{x_h}{x_t}{N_2}}(n)$:
\[ P_M(n) + (\Crec + P_{N_1}(n)) + n (\Cvar + \Crec + P_{N_2}(n)) \]
Instead of directly dealing with this bound on the cost, we use a stronger one:
\[ P_M(n) + (\Crec + P_{N_1}(n)) + \mathsf{length}(\ell) (\Cvar + \Crec + P_{N_2}(n)) \]
This bound makes intuitive sense; the whole point of multiplying by $n$ there was because it was suppose to bound the maximum number of times an iteration could run. Now that we are in the midst of the proof with access to dynamic information, we know we will iterate exactly $\mathsf{length}(\ell)$ times. This bound is indeed stronger; by \pref{thm:non-size-increasing}, we know $\mathsf{length}(\ell) \le \envsize{\env_1} \le n$.

From here, the strategy is to proceed by induction on the structure of $\ell$ and prove termination alongside this tighter bound. The cases of this induction are relatively straightforward and do not require anything more complicated than what is seen in the application case above. The base case makes use of the induction hypothesis for $N_1$, and the inductive case makes use of the induction hypothesis for $N_2$.
\end{proof}

\fi

\end{document}